\documentclass[twocolumn,secnumarabic,amssymb, nobibnotes, aps, pra, superscriptaddress]{revtex4-2}

\usepackage{amsmath,mathrsfs,amssymb,amsfonts,amsthm,mathtools}
\usepackage{graphicx,color}
\usepackage{subcaption}
\usepackage{booktabs}
\usepackage{multirow}
\usepackage{physics,color}
\usepackage{subcaption}
\usepackage{tabularx}
\usepackage{diagbox}
\usepackage[table]{xcolor}
\usepackage{tikz,tikz-cd}
\usetikzlibrary{shapes.geometric, arrows, positioning}
\usepackage{empheq}
\usepackage{bm}
\newtheorem{theorem}{Theorem}

\newtheorem{remark}{Remark}
\usepackage{physics}
\usepackage[T1]{fontenc}
\usepackage{bm}
\usepackage{empheq}
\usepackage{hyperref}
\hypersetup{
    colorlinks=true,
    linkcolor=blue,
    citecolor=blue,
    urlcolor=blue
}

\newcommand{\nn}{\nonumber}

\begin{document}
\title{Tunable Non-Gaussianity and Exact Higher-Order Coherences for Quantum Advantage
}


\author{Arash Azizi}
	\affiliation{Texas A\&M University, College Station, TX 77843}

\begin{abstract}
Non-Gaussian states are essential for achieving a quantum advantage in continuous-variable (CV) information processing. Among these, coherent superpositions of squeezed states are a foundational resource. While exact higher-order statistics are available in the undisplaced case, a complete and analytically tractable treatment \emph{with a common displacement} has been missing. We introduce and solve the displaced Janus state—a coherent superposition of two squeezed coherent states that share the same displacement—and develop an analytical framework, based on a family of Generalized Squeezing Polynomials, that yields closed-form expressions for arbitrary-order factorial moments and coherence functions \(g^{(k)}(0)\), the full Wigner function, and the quantum Fisher information. The analysis shows how interference at a fixed mean, driven by a mismatch of the component covariances rather than by mean separation, can be precisely engineered to transform the extreme photon bunching of the constituents into strong sub-Poissonian and even perfect multiphoton suppression. We further provide a rigorous quantum Fisher information analysis, proving that parameters encoded by linear generators (for example, the number operator) are bounded by the standard quantum limit, whereas parameters encoded by quadratic generators (for example, squeezing transformations) achieve Heisenberg-limited scaling. Together, these results furnish a complete analytical toolkit for a versatile class of non-Gaussian states, establishing the displaced Janus state as a key primitive for hybrid quantum protocols, quantum metrology, and fault-tolerant continuous-variable computation.
\end{abstract}

\maketitle

\section{Introduction}

Continuous–variable (CV) quantum optics provides a powerful platform for quantum information, with Gaussian states and operations enabling scalable protocols in computation and sensing~\cite{Mandel_Wolf1995,Scully_Zubairy1997,Schleich2001Phase_space,Barnett_Radmore2002,weedbrook2012gaussian,braunstein2005quantum,Agarwal2012book}. However, purely Gaussian resources are efficiently classically simulable and therefore cannot yield universal quantum advantage~\cite{ohliger2010limitations,Mari2012Wigner,Veitch2012}. Overcoming this limitation requires non-Gaussian ingredients, central to hybrid continuous–variable/discrete–variable (CV/DV) approaches~\cite{Andersen2015Hybrid}, where non-Gaussianity—often correlated with phase-space negativity—provides computational “magic” and metrological gains~\cite{dodonov2002nonclassical,Kenfack2004,Mari2012Wigner,Veitch2012,GenoniParis2010,Albarelli2018Resource,TakagiZhuang2018,Chabaud2021,Chabaud2024phasespace}.

Important non-Gaussian resources include photon-added and photon-subtracted states—classic photon-added coherent states already exhibit sub-Poissonian statistics, a singular Glauber–Sudarshan \(P\) function, and Wigner-function negativity~\cite{Agarwal_Tara1991,Zavatta2004,Parigi2007}—as well as Schr\"odinger-cat states~\cite{Ourjoumtsev2007,Vlastakis2013} and Gottesman–Kitaev–Preskill (GKP) grid states for fault tolerance~\cite{GKP2001,CampagneIbarcq2020,Menicucci2006,Gu2009CV,Andersen2015Hybrid}.

A foundational route to engineering non-Gaussianity is the coherent superposition of Gaussian states. For squeezed states, this direction traces to the seminal work by Sanders on superposed squeezed vacua~\cite{Sanders1989Superposition}, with properties mapped in subsequent work~\cite{Obada1997,Othman2024} and connections made to generalized coherent-state frameworks~\cite{Fox1999Generalized,Philbin2014AmJPhys}. While the general mathematical form of a superposition of two arbitrary squeezed coherent states was later considered by Barbosa et al.~\cite{Barbosa2000Generalized}, a systematic, arbitrary-order analytical framework for calculating their full photon statistics has remained elusive. Instead, much of the rich body of research has explored parallel or more restricted scenarios. This includes extensive studies on superpositions of \emph{coherent} states~\cite{Buzek1992Superpositions,Prakash2004Maximum} (synthesized in~\cite{Sanders2012Review}) and structurally different states like ``squeezed cat states'' (which have identical squeezing but opposite displacements)~\cite{Xin1994Even}. Other practical routes to non-Gaussianity and related resources have also been investigated, including photon-added variants and oscillator control protocols~\cite{Bohloul2024,Dehghani2019Excitation}, two-mode cat-like states~\cite{Oudot2015}, and bright squeezed vacuum~\cite{Iskhakov2012Superbunched}. Even when the correct state structure was considered, analyses have typically been limited to low-order statistical indicators like Mandel’s \(Q\) and \(g^{(2)}(0)\) without providing a general method for higher-order moments~\cite{Prakash2004Maximum,Guo2022HighOrder}.

\begin{table*}[t]
\centering
\footnotesize
\setlength{\tabcolsep}{4pt}
\renewcommand{\arraystretch}{1.12}
\caption{Comparison of cat, squeezed–cat (SC), and displaced Janus (DJ) states.}
\label{tab:state_comparison}
\begin{tabular}{lccc}
\hline
& Cat & SC & DJ \\
\hline
Component covariances      & identical (vacuum)        & identical (same squeeze)      & generally unequal ($V_\xi\neq V_\zeta$) \\
Component means            & $\pm\alpha$               & $\pm\alpha$                   & common mean $\alpha$ \\
Fringe driver              & mean separation           & mean separation               & covariance mismatch \\
Negativity at $|\alpha|=0$ & no                        & no                            & yes (unequal Gaussians) \\
$\min g^{(2)}(0)$, small $\bar n$ & $0$ (odd)        & $0$ (odd)                     & $\tfrac12$ (antisymmetric, $\alpha=0$) \\
$g^{(k>2)}(0)$ control     & limited                   & limited                       & can approach $0$ \\
Gaussian–unitary equivalent to cat & n/a              & yes ($S(\xi)\ket{C_\pm}$)     & only if $V_\xi=V_\zeta$ \\
\hline
\end{tabular}
\end{table*}

The urgency for a complete theory is now underscored by recent experimental paper demonstrating the physical generation of arbitrary superpositions of nonclassical states—including differently squeezed states—in controllable platforms like trapped ions~\cite{Saner2024generating}. To begin filling the theoretical gap, we recently introduced the \emph{Janus state}—a coherent superposition of two squeezed vacuum states with arbitrary squeezing and superposition coefficients~\cite{Azizi2025Janus}. Named by analogy with the two-faced Roman deity, our analysis revealed that the superposition can be engineered to show strong antibunching, reaching the universal lower bound of $g^{(2)}(0) = 1/2$. A follow-up investigation of the state's higher-order coherences revealed that interference can be tuned to perfectly suppress multi-photon events ($g^{(k>2)} \to 0$), a phenomenon we term \emph{complementarity ad infinitum}~\cite{Azizi2025Janus_higher}. The key to this arbitrary-order analysis was the development of a new mathematical framework built on a family of \emph{Generalized Squeezing Polynomials}.

However, these foundational studies were restricted to the case of zero displacement ($\alpha=0$). Building on this foundation, we introduce and provide a full analytical solution for the \textbf{displaced Janus state (DJ)}: a coherent superposition of two squeezed coherent states that share a common displacement but have fully independent squeezing parameters and superposition weights. Throughout, we adopt the “squeeze-then-displace” convention $|\alpha,\xi\rangle \equiv D(\alpha)S(\xi)\ket{0}$. The displaced Janus state is
\begin{align}
    |\Psi\rangle = D(\alpha) \Big( \chi\, S(\xi) |0\rangle + \eta\,S(\zeta) |0\rangle \Big),
    \label{eq:disp_Janus_State}
\end{align}
where $\chi, \eta, \xi$ and $\zeta$ are four general complex parameters. We provide closed-form expressions for cross moments and the Wigner function, and exact higher-order coherences, thereby generalizing earlier zero-displacement or special-symmetry cases. A striking contrast emerges: whereas a single squeezed vacuum exhibits super-Poissonian correlations (e.g., $g^{(2)}(0)=3+1/\bar n\!\to\!3$, $g^{(3)}(0)\!\to\!15$, $g^{(4)}(0)\!\to\!105$ as $\bar n\!\to\!\infty$, while all $g^{(k)}(0)$ diverge as $\bar n\!\to\!0$), the Janus superposition can approach the two-photon Fock limit with $g^{(2)}(0)=\tfrac12$ and $g^{(k)}(0)\to 0$ for $k\ge 3$. 
As we show, adding a common displacement invariably increases the second-order coherence above this floor in a controlled, analytical way, with the minimum $g^{(2)}(0)=\tfrac{1}{2}$ only being achievable at zero displacement.

Higher-order correlation functions serve as practical diagnostics of non-Gaussian structure. While negativity of a Wigner distribution is a sufficient (not necessary) indicator of non-Gaussianity, Gaussian states are entirely determined by their first and second moments; thus any measured deviation of moments of order $k\ge 3$ from the Gaussian predictions (given the observed mean and covariance) certifies non-Gaussian features. The closed-form expressions derived here provide a measurable blueprint of the photon-number distribution, directly compatible with photon-number-resolving detection.

\noindent\emph{Summary (see Table~\ref{tab:state_comparison}).} Cats and squeezed-cats generate fringes by separating the component \emph{means}. The displaced Janus state (DJ) fixes a \emph{common mean} and drives fringes through a \emph{covariance mismatch}, yielding Wigner negativity even at $|\alpha|=0$ and enabling strong, order-selective suppression of higher-order correlations. The DJ reduces to a squeezed-cat when $\xi=\zeta$ and $(\chi,\eta)=(1,\pm 1)$; otherwise it is intrinsically non-Gaussian.

\medskip
\noindent\textit{Contributions.} We develop an analytical toolkit that
\begin{itemize}
 \item introduces a family of Generalized Squeezing Polynomials (GSPs) to compute factorial moments to arbitrary order;
 \item derives closed-form off-diagonal moments $\langle \alpha,\zeta|\,a^{\dagger k} a^{k}\,|\alpha,\xi\rangle$ that isolate interference effects;
 \item constructs full Wigner-function representations with explicit interference-induced negativity regions;
 \item explains photon-statistics behavior, including exact cancellation of \(O(r^{2})\) terms in \(g^{(k)}(0)\) in an antisymmetric small-squeezing limit;
 \item and identifies regimes achieving Heisenberg-limited quantum Fisher information (QFI) for nonlinear parameters (with limits approaching \(D(\alpha)\ket{2}\)).
\end{itemize}

Taken together, these results establish the displaced Janus state as a tunable, non-Gaussian resource that exemplifies hybrid CV/DV principles, enables metrological gains (e.g., in SU(1,1) interferometry), and offers analytic handles for resource quantification~\cite{Yurke1986,Hudelist2014,Giovannetti2006,Pezzesmerzi2018RMP,GenoniParis2010,Albarelli2018Resource}. They generalize our recent findings~\cite{Azizi2025Janus,Azizi2025Janus_higher} and furnish a reusable framework for designing and benchmarking non-Gaussian states in realistic CV platforms.

The remainder of the paper is organized as follows. Sec.~\ref{sec:fac_moments} develops the operator–displacement rearrangement and Fock–basis machinery that reduce all off–diagonal matrix elements to a universal series. 
Sec.~\ref{sec:GSP} introduces the Generalized Squeezing Functions and Polynomials, derives their symmetry and recurrence relations, and tabulates $P_{p,q}(z)$ up to order ten. 
Sec.~\ref{sec:Photon_Statistics} assembles closed-form expressions for $\mathcal{N}_k$ and $g^{(k)}(0)$ for a single squeezed coherent state and for the displaced Janus superposition, including the $\alpha\!\to\!0$ Janus limit and small-squeezing expansions. 
Sec.~\ref{sec:Non-Gaussian} characterizes the displaced Janus state via its (cross-)Wigner functions and analyzes quantum Fisher information (QFI) for linear and quadratic generators, highlighting Standard Quantum Limit (SQL) vs.\ Heisenberg scaling. 
Sec.~\ref{sec:Exp_feas} details optical platforms and state preparation, loss maps and robustness, balanced-homodyne tomography, and access to higher-order correlations with time-multiplexed and photon-number-resolving detection. 
Sec.~\ref{sec:conc} concludes with outlook. 
Appendices~\ref{app:P_prop} and \ref{app:low_M} provide numerics for $F_{p,q}$ and constructive algorithms for $P_{p,q}(z)$, and give step-by-step derivations of the low-order matrix elements used in the main text.

\section{Calculation of the Factorial Moments} \label{sec:fac_moments}

\subsection{General Framework}

We begin by defining the core components of our state. The squeezing operator is defined as:
\begin{align}
    S(\xi) = \exp\left( \frac{1}{2} (\xi^* a^2 - \xi a^{\dagger 2}) \right)
\end{align}
where the complex parameter $\xi = r e^{i\theta}$ defines the squeezing strength ($r$) and angle ($\theta$).

A squeezed coherent state is created by applying both squeezing and displacement to the vacuum. There are two conventions for the order of operations: squeezing the vacuum then displacing ($D(\alpha)S(\xi)|0\rangle$), or displacing the vacuum then squeezing ($S(\xi)D(\alpha)|0\rangle$). As is well known, the squeezing–displacement operator identity~\cite{Barnett_Radmore2002,Agarwal2012book} states
\begin{align}
    S(\xi)D(\alpha) =& D(\alpha')S(\xi), 
    \quad \text{where} \nn\\
    \alpha' =& \alpha \cosh r - \alpha^* \sinh r\, e^{i\theta}.
\end{align}
This relation shows that squeezing a coherent state is equivalent to displacing a squeezed vacuum state, with the displacement amplitude modified according to $\alpha'$. For calculational convenience throughout this work, we adopt the ``squeeze-then-displace'' convention:
\begin{align}
    |\alpha, \xi \rangle = D(\alpha) S(\xi) |0\rangle.
\end{align}

Our state of interest, the displaced Janus state, is a general superposition of two such states sharing the same displacement:
\begin{align}
    |\Psi\rangle = \chi\,|\alpha, \xi \rangle + \eta\,|\alpha, \zeta \rangle,
    \label{eq:Janus_State_def}
\end{align}
where $\zeta = s e^{i\phi}$ is the squeezing parameter of the second component. To characterize the photon statistics of $|\Psi\rangle$, we calculate the $k$-th factorial moment, defined by the normally ordered expectation value:
\begin{align}
    \mathcal{N}_k \equiv \langle \Psi | a^{\dagger k} a^k | \Psi \rangle.
\end{align}
Substituting Eq.~\eqref{eq:Janus_State_def} into this definition yields three distinct contributions:
\begin{align}
    \mathcal{N}_k
    &= |\chi|^2\,\langle  \alpha, \xi | a^{\dagger k} a^k \ket{\alpha, \xi}
    + |\eta|^2\,\langle \alpha , \zeta| a^{\dagger k} a^k 
    | \alpha, \zeta \rangle \nn\\
    &\quad+ 2\,\mathrm{Re}\!\left[ \eta^*\,\chi \, \langle \alpha , \zeta| a^{\dagger k} a^k \ket{\alpha, \xi}\,\right].
    \label{eq:Nk_expanded_full}
\end{align}
The first two terms are the diagonal moments of the individual component states, while the third term encodes the quantum interference between them via the off-diagonal matrix element. Evaluating this general off-diagonal term is the central challenge addressed in the following sections.

\subsection{Transforming the Matrix Element}
A powerful way to evaluate such off-diagonal elements is to use the transformation properties of the displacement operator, shifting the displacement action from the states onto the operators. Since both components of $|\Psi\rangle$ share the same displacement $\alpha$, the off-diagonal term can be written compactly as
\begin{align}
    \mathcal{M}_k \equiv \langle \alpha , \zeta| a^{\dagger k} a^k \ket{\alpha, \xi},
\end{align}
whose evaluation is central to our analysis.

Using the definition of the states, we can write this as:
\begin{align}
    \mathcal{M}_k = \langle 0 | S^\dagger(\zeta) D^\dagger(\alpha) a^{\dagger k} a^k D(\alpha) S(\xi) | 0 \rangle.
\end{align}
We can simplify the central operator block, $D^\dagger(\alpha) a^{\dagger k} a^k D(\alpha)$, by repeatedly using the well-known identities $D^\dagger(\alpha) a D(\alpha) = a + \alpha$ and $D^\dagger(\alpha) a^\dagger D(\alpha) = a^\dagger + \alpha^*$. This effectively removes the displacement operators at the cost of modifying the creation and annihilation operators:
\begin{align}
    D^\dagger(\alpha) a^{\dagger k} a^k D(\alpha) =& \big(D^\dagger(\alpha) a^\dagger D(\alpha)\big)^k \big(D^\dagger(\alpha) a D(\alpha)\big)^k \nn\\
    =& (a^\dagger + \alpha^*)^k (a + \alpha)^k.
\end{align}
The matrix element is thereby transformed into an expectation value with respect to the non-displaced squeezed vacuum states $|\xi\rangle$ and $|\zeta\rangle$:
\begin{align}
    \mathcal{M}_k = \langle \zeta | (a^\dagger + \alpha^*)^k (a + \alpha)^k | \xi \rangle.
    \label{eq:transformed_matrix_element}
\end{align}
This transformation is the key to the entire calculation, as the problem is now reduced to evaluating the action of polynomial operators on squeezed vacuum states.

\subsection{Fock Basis Calculation}

To evaluate Eq.~\eqref{eq:transformed_matrix_element}, we expand the operators and states in the Fock basis. First, we expand the operator terms using the binomial theorem:
\begin{align}
    (a + \alpha)^k =& \sum_{p=0}^k \binom{k}{p} \alpha^{k-p} a^p \quad \text{and} \nn\\
    (a^\dagger + \alpha^*)^k =& \sum_{q=0}^k \binom{k}{q} {\alpha^*}^{k-q} a^{\dagger q}.
\end{align}
The action of the annihilation operator $a^p$ on a squeezed vacuum state $|\xi\rangle$ is given by:
\begin{align}
    a^p |\xi\rangle =& (\cosh r)^{-1/2} \\
    &\times
    \sum_{m=\lceil p/2 \rceil}^\infty (\tanh r\, e^{i\theta})^m \frac{(2m-1)!!}{\sqrt{(2m-p)!}} |2m - p\rangle. \nn
\end{align}
Combining these expansions, the matrix element becomes a detailed sum over the Fock basis states:
\begin{widetext}
\begin{align}
    \mathcal{M}_k &= \sum_{p=0}^k \sum_{q=0}^k \binom{k}{p} \binom{k}{q} \alpha^{k-p} {\alpha^*}^{k-q} \langle \zeta | a^{\dagger q} a^p | \xi \rangle \notag \\
    &= (\cosh r \cosh s)^{-1/2} \sum_{p=0}^k \sum_{q=0}^k \binom{k}{p} \binom{k}{q} \alpha^{k-p} {\alpha^*}^{k-q} \notag \\
    &\quad \times \sum_{m,n=0}^\infty (\tanh r e^{i\theta})^m (\tanh s e^{-i\phi})^n \frac{(2m-1)!!}{\sqrt{(2m-p)!}} \frac{(2n-1)!!}{\sqrt{(2n-q)!}} \langle 2n-q | 2m-p \rangle.
\end{align}
The inner product of the Fock states, $\langle 2n - q | 2m - p \rangle$, acts as a Kronecker delta, $\delta_{2n-q, 2m-p}$. This enforces the condition $2n - q = 2m - p$, which allows us to eliminate one of the summation variables, say $n$:
\begin{align}
    n = m + \frac{1}{2}(q - p).
\end{align}
For $n$ to be an integer, the difference $q-p$ must be an even number, a condition that is naturally satisfied as the squeezed states are superpositions of only even-numbered Fock states. After substituting for $n$ and performing some algebraic rearrangement of the factorial terms, the full matrix element becomes:
\begin{align}
\mathcal{M}_k
&= \frac{1}{\sqrt{\cosh r\,\cosh s}}
\sum_{\substack{0 \le p,q \le k \\ p \equiv q \ (\mathrm{mod}\ 2)}}
\binom{k}{p} \binom{k}{q} \,
\alpha^{\,k-p} \, (\alpha^*)^{\,k-q} \,
\big(\tanh s\, e^{-i\phi}\big)^{\frac{q-p}{2}}
\nonumber \\
&\quad\times
\sum_{m = \lceil \frac{p}2 \rceil}^{\infty}
\frac{(2m)!}{(2m-p)!} \,
\frac{(2m+q-p-1)!!}{(2m)!!} \,
\Big[ \tanh r \, e^{i\theta} \, \tanh s \, e^{-i\phi} \Big]^{m}.
\label{eq:Mk_parity}
\end{align}
The matrix element becomes a detailed sum over the Fock basis states. This recurring structure motivates the definition of a new family of special functions.

The structure of the inner sum over the index $m$ is universal and depends only on the indices $p$ and $q$ and the composite parameter 
\begin{align}
    z = \tanh r \tanh s\, e^{i(\theta-\phi)}. \label{z}
\end{align}
This recurring structure motivates the definition of a new family of special functions, which we introduce now to simplify the result.

\section{Generalized Squeezing Polynomials} \label{sec:GSP}

We define the Generalized Squeezing Function, $F_{p,q}(z)$, as the recurring infinite series that appears when evaluating the matrix elements in the Fock basis:
\begin{align}
F_{p,q}(z)
\equiv &
\sum_{n=\big\lceil\tfrac{p}{2}\big\rceil}^{\infty}
\frac{(2n)!}{(2n-p)!} \,
\frac{(2n+q-p-1)!!}{(2n)!!} \,
z^{\,n}, 
\qquad p \equiv q \ (\mathrm{mod}\ 2).
\label{eq:Fpq_def}
\end{align}

The indices $p$ and $q$ must have the same parity for the function to be non-zero, a natural consequence of the Fock space calculation.

With this definition, the full off-diagonal matrix element $\mathcal{M}_k = \langle \zeta | (a^\dagger + \alpha^*)^k (a + \alpha)^k | \xi \rangle$ can be expressed in a much more compact and organized form. By substituting the definition of $F_{p,q}(z)$ back into the full expansion, we arrive at the final result:
\begin{align}
\mathcal{M}_k
=& \langle  \alpha, \zeta | a^{\dagger k} a^k \ket{\alpha, \xi} \nn\\
=& \frac{1}{\sqrt{\cosh r\,\cosh s}}
\sum_{\substack{0 \le p,q \le k \\[1pt] p \equiv q \ (\mathrm{mod}\ 2)}}
\binom{k}{p}\binom{k}{q}\,
(\alpha^*)^{\,k-q}\,\alpha^{\,k-p}\,
\big(\tanh s\,e^{-i\phi}\big)^{\frac{q-p}{2}}\,
F_{p,q}(z).
\label{eq:Mk_final_compact}
\end{align}
This expression is the central tool for calculating the interference terms in the photon statistics of the superposition state. The complexity of the Fock-space expansion is now neatly encapsulated within the special function $F_{p,q}(z)$, which is the subject of the next section.

Representative behaviors of \(F_{p,q}(z)\) for real \(z\) and their algebraic singularity at \(z\!\to\!1^{-}\) are shown in Fig.~\ref{fig:Fpq_plots}.

\begin{figure*}[t!]
  \centering
  \includegraphics[width=.8\textwidth]{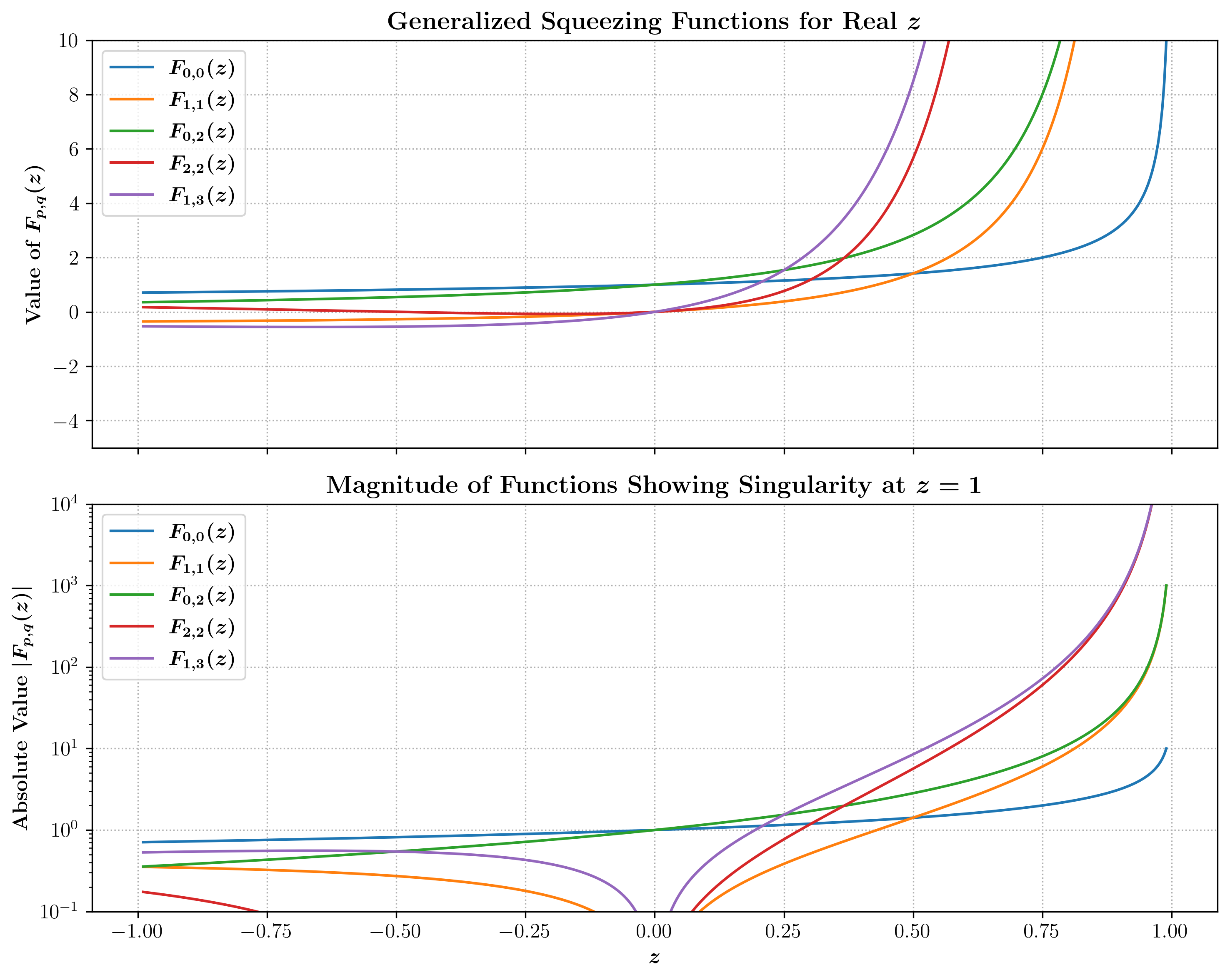}
  \caption{Generalized Squeezing Functions \(F_{p,q}(z)\) for real \(z\in[-1,1)\).
  Top: linear scale. Bottom: log magnitude highlighting the algebraic
  singularity at \(z\to 1^{-}\) predicted by \(F_{p,q}(z)=P_{p,q}(z)/(1-z)^{(p+q+1)/2}\).
  Curves shown: \(F_{0,0}, F_{1,1}, F_{0,2}, F_{2,2}, F_{1,3}\).}
  \label{fig:Fpq_plots}
\end{figure*}

\subsection{Generalized Squeezing Polynomials and Recurrence Relations}

The singular behavior of the function $F_{p,q}(z)$ as $z \to 1$ can be factored out to define the Generalized Squeezing Polynomial, $P_{p,q}(z)$:
\begin{align}
F_{p,q}(z) \equiv \frac{P_{p,q}(z)}{(1-z)^{(p+q+1)/2}}.
\label{eq:Ppq_def}
\end{align}
These generalized polynomials satisfy a set of recurrence relations that allow for their systematic construction. Two fundamental relations can be derived directly from the definition of $F_{p,q}(z)$:
\begin{align}
P_{p+1,q+1}(z) &= \Big( (2p+q+1)z - p \Big) P_{p,q}(z) + 2z(1-z) P'_{p,q}(z), \label{eq:Poly_Rec1} \\
P_{p,q+2}(z) &= \Big( 2pz - p+q+1 \Big)P_{p,q}(z) + 2z(1-z)P'_{p,q}(z). \label{eq:Poly_Rec2}
\end{align}
While these relations are fundamental, they are not sufficient on their own to generate polynomials where $p > q$. By algebraically combining them, a third, three-term recurrence relation can be found that allows for stepping the $p$ index:
\begin{align}
P_{p+2,q}(z) = P_{p+1,q+1}(z) + q(z-1)P_{p+1,q-1}(z). \label{eq:Poly_Rec3}
\end{align}
Together, this complete set of relations allows for the generation of any polynomial starting from the initial condition $P_{0,0}(z)=1$. For example, after generating the $p=0$ and $p=1$ rows, Eq.~\eqref{eq:Poly_Rec3} can be used to construct all subsequent rows of the polynomial table.

\subsection{Explicit Forms of Generalized Squeezing Polynomials}
The explicit forms of the Generalized Squeezing Polynomials, $P_{p,q}(z)$, can be generated systematically from the initial condition $P_{0,0}(z)=1$ using the recurrence relations previously defined in Eqs.~\eqref{eq:Poly_Rec1}-\eqref{eq:Poly_Rec3}.

A key feature of this framework is that it provides a natural extension of the standard, single-mode squeezing polynomials. The polynomials that lie on the main diagonal of the table, where $p=q=k$, are identical to the standard squeezing polynomials, $P_k(z)$. The off-diagonal terms ($p \neq q$) then describe the more complex correlations that arise from the two-mode interference central to the Janus state, as explored in Ref.~\cite{Azizi2025Janus_higher}.

Table~\ref{tab:Ppq_even_even} and Table~\ref{tab:Ppq_odd_odd} present these polynomials for cases up to $p,q \leq 10$, organizing them by their respective indices.

\begin{table*}[t]
\centering
\caption{Generalized Squeezing Polynomials $P_{p,q}(z)$ for even-even indices ($0\le p,q\le 10$). The omitted diagonal element is $P_{10,10}(z) = 3628800z^{10}+81648000z^{9}+285768000z^{8}+238140000z^{7}+44651250z^{6}+893025z^{5}$. The other elements of the $q=10$ column can be found using the symmetry relation $P_{p,10}(z)=z^{(p-10)/2}P_{10,p}(z)$.}
\label{tab:Ppq_even_even}
\renewcommand{\arraystretch}{1.35}
\resizebox{\columnwidth}{!}{%
\begin{tabular}{|c||c|c|c|c|c|}
\hline
$p \backslash q$ & \textbf{0} & \textbf{2} & \textbf{4} & \textbf{6} & \textbf{8} \\
\hline\hline
\textbf{0}  & $1$
            & $1$
            & $3$
            & $15$
            & $105$ \\
\hline
\textbf{2}  & $z$
            & $2z^{2}+z$
            & $12z^{2}+3z$
            & $90z^{2}+15z$
            & $840z^{2}+105z$ \\
\hline
\textbf{4}  & $3z^{2}$
            & $12z^{3}+3z^{2}$
            & $24z^{4}+72z^{3}+9z^{2}$
            & $360z^{4}+540z^{3}+45z^{2}$
            & $5040z^{4}+5040z^{3}+315z^{2}$ \\
\hline
\textbf{6}  & $15z^{3}$
            & $90z^{4}+15z^{3}$
            & $360z^{5}+540z^{4}+45z^{3}$
            & $720z^{6}+5400z^{5}+4050z^{4}+225z^{3}$
            & $20160z^{6}+75600z^{5}+37800z^{4}+1575z^{3}$ \\
\hline
\textbf{8}  & $105z^{4}$
            & $840z^{5}+105z^{4}$
            & $5040z^{6}+5040z^{5}+315z^{4}$
            & $20160z^{7}+75600z^{6}+37800z^{5}+1575z^{4}$
            & $40320z^{8}+564480z^{7}+1058400z^{6}+352800z^{5}+11025z^{4}$ \\
\hline
\textbf{10} & $945z^{5}$
            & $9450z^{6}+945z^{5}$
            & $75600z^{7}+56700z^{6}+2835z^{5}$
            & $453600z^{8}+1134000z^{7}+425250z^{6}+14175z^{5}$
            & $1814400z^{9}+12700800z^{8}+15876000z^{7}+3969000z^{6}+99225z^{5}$ \\
\hline
\end{tabular}%
}
\end{table*}

\begin{table*}[t]
\centering
\caption{Generalized Squeezing Polynomials $P_{p,q}(z)$ for odd-odd indices ($1\le p,q\le 9$). The omitted diagonal element is $P_{9,9}(z) = 362880z^{9}+6531840z^{8}+17146080z^{7}+9525600z^{6}+893025z^{5}$. The other elements of the $q=9$ column can be found using the symmetry relation $P_{p,9}(z)=z^{(p-9)/2}P_{9,p}(z)$.}
\label{tab:Ppq_odd_odd}
\renewcommand{\arraystretch}{1.35}
\begin{tabular}{|c||c|c|c|c|}
\hline
$p \backslash q$ & \textbf{1} & \textbf{3} & \textbf{5} & \textbf{7} \\
\hline\hline
\textbf{1} & $z$
           & $3z$
           & $15z$
           & $105z$ \\
\hline
\textbf{3} & $3z^{2}$
           & $6z^{3}+9z^{2}$
           & $60z^{3}+45z^{2}$
           & $630z^{3}+315z^{2}$ \\
\hline
\textbf{5} & $15z^{3}$
           & $60z^{4}+45z^{3}$
           & $120z^{5}+600z^{4}+225z^{3}$
           & $2520z^{5}+6300z^{4}+1575z^{3}$ \\
\hline
\textbf{7} & $105z^{4}$
           & $630z^{5}+315z^{4}$
           & $2520z^{6}+6300z^{5}+1575z^{4}$
           & $5040z^{7}+52920z^{6}+66150z^{5}+11025z^{4}$ \\
\hline
\textbf{9} & $945z^{5}$
           & $7560z^{6}+2835z^{5}$
           & $45360z^{7}+75600z^{6}+14175z^{5}$
           & $181440z^{8}+952560z^{7}+793800z^{6}+99225z^{5}$ \\
\hline
\end{tabular}
\end{table*}

\subsection{Symmetry Relation}

\begin{theorem}[Symmetry of Generalized Squeezing Polynomials]
The Generalized Squeezing Polynomials obey the following symmetry relation for all non-negative integers $p$ and $q$ of the same parity:
\begin{align}
    P_{p,q}(z) = z^{(p-q)/2} P_{q,p}(z).
    \label{eq:Poly_Symm}
\end{align}
\end{theorem}
\begin{proof}
The theorem for the polynomials, $P_{p,q}(z) = z^{(p-q)/2} P_{q,p}(z)$, is a direct consequence of an equivalent symmetry for their generating functions, $F_{p,q}(z)$. We will prove the relation:
\begin{equation*}
    F_{p,q}(z) = z^{(p-q)/2} F_{q,p}(z).
\end{equation*}
We start with the right-hand side and show that it transforms into the left-hand side. By definition,
\begin{equation*}
    z^{(p-q)/2} F_{q,p}(z) = z^{(p-q)/2} \sum_{m=0}^{\infty}  \frac{(2m)!}{(2m-q)!} \frac{(2m+p-q-1)!!}{(2m)!!} z^m,
\end{equation*}
where the summation is over the appropriate range of integers $m$. Combining the powers of $z$ yields:
\begin{equation*}
    \sum_{m=0}^{\infty} \frac{(2m)!}{(2m-q)!} \frac{(2m+p-q-1)!!}{(2m)!!} z^{m + (p-q)/2}.
\end{equation*}
We now perform a change of summation variable. Since $p$ and $q$ have the same parity, their difference $p-q$ is an even integer, which means $(p-q)/2$ is an integer. Let the new summation index be $n = m + (p-q)/2$, which implies $m = n - (p-q)/2$. This transforms the power of $z$ to $z^n$.

The core of the proof is showing that the transformed coefficient is identical to the original coefficient for $F_{p,q}(z)$. We must prove the equality:
\begin{equation*}
    \frac{(2n-p+q)!}{(2n-p)!} \frac{(2n-1)!!}{(2n-p+q)!!} = \frac{(2n)!}{(2n-p)!} \frac{(2n+q-p-1)!!}{(2n)!!}.
\end{equation*}
After canceling the common denominator $(2n-p)!$ and applying the identity $(2k)!/(2k)!! = (2k-1)!!$ to both sides, the expression reduces to the trivial identity $(2n-p+q-1)!! = (2n+q-p-1)!!$.

Since the coefficients match for every $n$, the series are identical. This completes the proof.
\end{proof}

\noindent
This symmetry property significantly simplifies the construction of the polynomial table. One only needs to compute the values for the upper triangle (where $q \ge p$); the lower triangle can then be generated instantly by this rule.

\subsection{Explicit Results for Low-Order Matrix Elements}

To demonstrate the power of the general formalism, we now obtain explicit expressions for the first few matrix elements $\mathcal{M}_k$, which serve as building blocks for state normalization and the coherence functions $g^{(n)}(0)$. Throughout we write the displacement as $\alpha = |\alpha|e^{i\varphi}$. Detailed derivations are provided in Appendix~\ref{app:low_M}.

\subsubsection*{\texorpdfstring{Zeroth Order ($k=0$): State Overlap}{}}

For $k=0$, the general formula for $\mathcal{M}_k$ reduces to a single term ($p=q=0$), yielding the state overlap $\mathcal{M}_0 = \langle  \alpha ,\zeta \ket{\alpha, \xi}$:
\begin{align}
    \mathcal{M}_0 = \frac{1}{\sqrt{\cosh r \cosh s - \sinh r \sinh s\, e^{i(\theta - \phi)}}}. \label{eq:M0}
\end{align}

\subsubsection*{\texorpdfstring{First Order ($k=1$)}{}}

For $k=1$, the sum consists of two terms ($p=q=0$ and $p=q=1$). The matrix element $\mathcal{M}_1 = \langle \alpha , \zeta| a^\dagger a \ket{\alpha, \xi}$, which contributes to the mean photon number, is found to be:
\begin{align}
    \mathcal{M}_1 = \frac{1}{\sqrt{\cosh r \cosh s}\,(1-z)^{3/2}} \Big[ |\alpha|^2 (1-z) + z \Big]. \label{eq:M1}
\end{align}

\subsubsection*{\texorpdfstring{Second Order ($k=2$)}{}}

The $k=2$ element is crucial for calculating the second-order coherence, $g^{(2)}(0)$. Carrying out the summation over the five allowed pairs of $(p,q)$ yields the final explicit result:
\begin{align}
\mathcal{M}_2 = \frac{1}{\sqrt{\cosh r \cosh s}\,(1-z)^{5/2}} \bigg[ & (1-z)^2 |\alpha|^4 + (2z^2+z) \label{eq:M2}\\
& + (1-z)|\alpha|^2 \Big( 4z + \tanh s\, e^{i(2\varphi - \phi)} + \tanh r\, e^{i(\theta - 2\varphi)} \Big) \bigg].  \nonumber
\end{align}

\subsubsection*{\texorpdfstring{Third Order ($k=3$)}{}}
Finally, the $k=3$ element, needed for $g^{(3)}(0)$, is found by summing over all eight allowed pairs of $(p,q)$. This procedure gives:
\begin{align}
\mathcal{M}_3 =& \frac{1}{\sqrt{\cosh r \cosh s}\,(1-z)^{7/2}} \times \Bigg[ 
 (1-z)^3 |\alpha|^6 
 + (6z^3+9z^2) \label{eq:M3}\\
& + (1-z)^2 |\alpha|^4 \Big( 9z + 3\tanh s\, e^{i(2\varphi - \phi)} + 3\tanh r\, e^{i(\theta - 2\varphi)} \Big) \nonumber \\
& + (1-z) |\alpha|^2 \Big( 9(2z^2+z) + 9z \big( \tanh s\, e^{i(2\varphi-\phi)} + \tanh r\, e^{i(\theta-2\varphi)} \big) \Big) \Bigg]
 \nonumber
\end{align}

\section{General-Order Photon Statistics} \label{sec:Photon_Statistics}

This section details the central results of our work: the derivation of the general-order correlation function, $g^{(k)}(0)$, for both a single squeezed coherent state and a superposition of two such states. The correlation function is defined as the ratio of the $k$-th factorial moment to the $k$-th power of the mean photon number:
\begin{align}
    g^{(k)}(0) = \frac{\langle a^{\dagger k} a^k \rangle}{\langle a^\dagger a \rangle^k} \equiv \frac{\mathcal{N}_k}{\mathcal{N}_1^k}.
\end{align}
We first analyze the fundamental case of a single state, and then use those results as building blocks to construct the solution for the displaced Janus state.
\subsection{Single Squeezed Coherent State}
The factorial moments $\mathcal{N}_k = \langle  \alpha, \xi | a^{\dagger k} a^k \ket{\alpha, \xi}$ for a single squeezed coherent state are obtained by specializing our general results for the off-diagonal matrix elements. This leads to a general expression for $\mathcal{N}_k$ which can be used to construct the coherence function $g^{(k)}(0) = \mathcal{N}_k / \mathcal{N}_1^k$. While the general formula is powerful, it is instructive to write out and visualize the explicit expressions for the first few orders.

\subsubsection*{Second-Order Coherence, \texorpdfstring{$g^{(2)}(0)$}{}}
The second-order coherence function is given by the complete analytical expression:
\begin{align}
g^{(2)}(0) = \frac{|\alpha|^4 + |\alpha|^2 \left( 4\sinh^2 r + 2\sinh r \cosh r \cos(2\varphi - \theta) \right) + 3\sinh^4 r + \sinh^2 r}{\left( |\alpha|^2 + \sinh^2 r \right)^2}. 
\label{eq:g2_single_state}
\end{align}

This general expression reproduces the expected limits: for \(\alpha=0\) it reduces to the squeezed vacuum, and for \(r=0\) it yields the coherent state. To chart the intermediate regime, we fix \(r=1\) and scan the displacement \(|\alpha|\) together with the relative phase \(2\varphi-\theta\). The \emph{parameter-space} maps in Fig.~\ref{fig:g2_rotated} show pronounced tunability: rotating the squeezing axis (changing \(\theta\)) rigidly rotates the domains of sub-Poissonian \((g^{(2)}(0)<1)\) and super-Poissonian statistics about the origin. The black contour marks the crossover \(g^{(2)}(0)=1\). For large \(|\alpha|\) all panels relax to the coherent-state benchmark \(g^{(2)}(0)\to 1\), whereas near the origin the depth and orientation of antibunching/bunching are set by the phase offset \(2\varphi-\theta\). Note that Fig.~\ref{fig:g2_rotated} (and likewise Figs.~\ref{fig:g3_rotated} and \ref{fig:g4_rotated}) shows \emph{parameter-space} maps of \(g^{(k)}(0)\) versus \(|\alpha|\) and \(2\varphi-\theta\), not phase-space quasiprobabilities (Wigner or \(Q\)); here \(\alpha\) is the common displacement amplitude (the coherent mean), not a point in \((q,p)\).

\begin{figure*}[htbp]
  \centering
  \includegraphics[width=.8\textwidth]{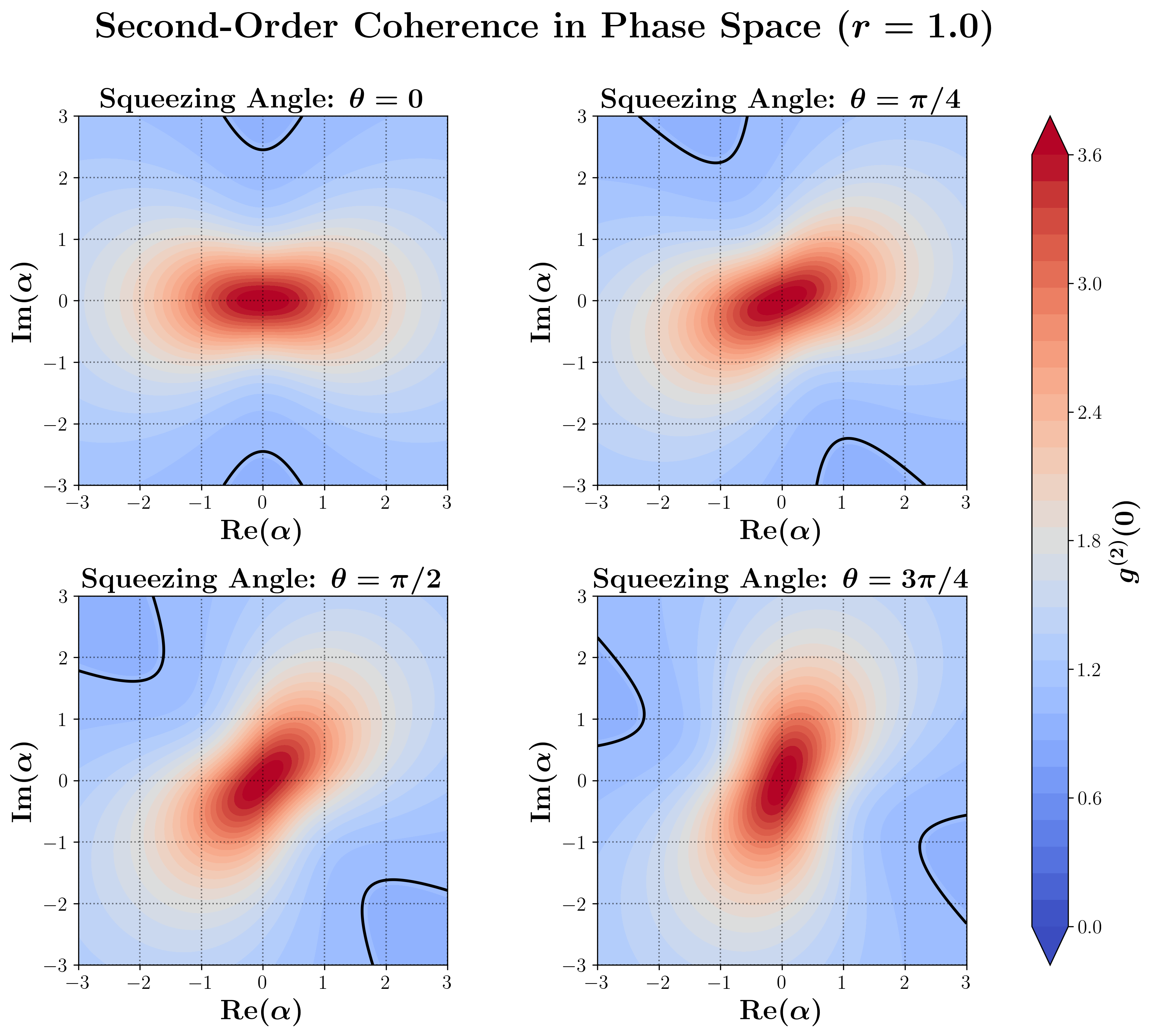}
  \caption{Phase–space maps of \(g^{(2)}(0)\) for a squeezed coherent state with \(r=1\). Each panel corresponds to a different squeezing angle \(\theta\); since the dependence is through the relative phase \(2\varphi-\theta\), varying \(\theta\) rotates the pattern. The black contour indicates \(g^{(2)}(0)=1\); blue (red) shades denote sub- (super-) Poissonian regions.}
  \label{fig:g2_rotated}
\end{figure*}

\subsubsection*{Third-Order Coherence, \texorpdfstring{$g^{(3)}(0)$}{}}

For $k=3$, the factorial moment $\mathcal{N}_3$ is found to be:

\begin{align}
\mathcal{N}_3 = &|\alpha|^6 + |\alpha|^4\Big(9\sinh^2 r + 6\sinh r \cosh r \cos(2\varphi-\theta)\Big) \nonumber\\
& + |\alpha|^2\Big(27\sinh^4 r + 9\sinh^2 r + 18\sinh^3 r \cosh r \cos(2\varphi-\theta)\Big) \nonumber \\
& + 15\sinh^6 r + 9\sinh^4 r.
\end{align}

The third-order coherence \(g^{(3)}(0)=\mathcal N_3/\mathcal N_1^{3}\) is shown on a logarithmic scale to match the \(g^{(4)}(0)\) panels (while \(g^{(2)}(0)\) remains linear). Figure~\ref{fig:g3_rotated} plots \(\log_{10} g^{(3)}(0)\) for representative \(\theta\), which avoids central saturation, highlights the rotational dependence, and makes multiplicative changes easy to compare. Since \(g^{(3)}(0)>0\), the logarithm is well defined; for numerical robustness we plot \(\log_{10}\!\big(g^{(3)}(0)+\epsilon\big)\) with \(\epsilon=10^{-8}\).

\begin{figure*}[t]
  \centering
  \includegraphics[width=.8\textwidth]{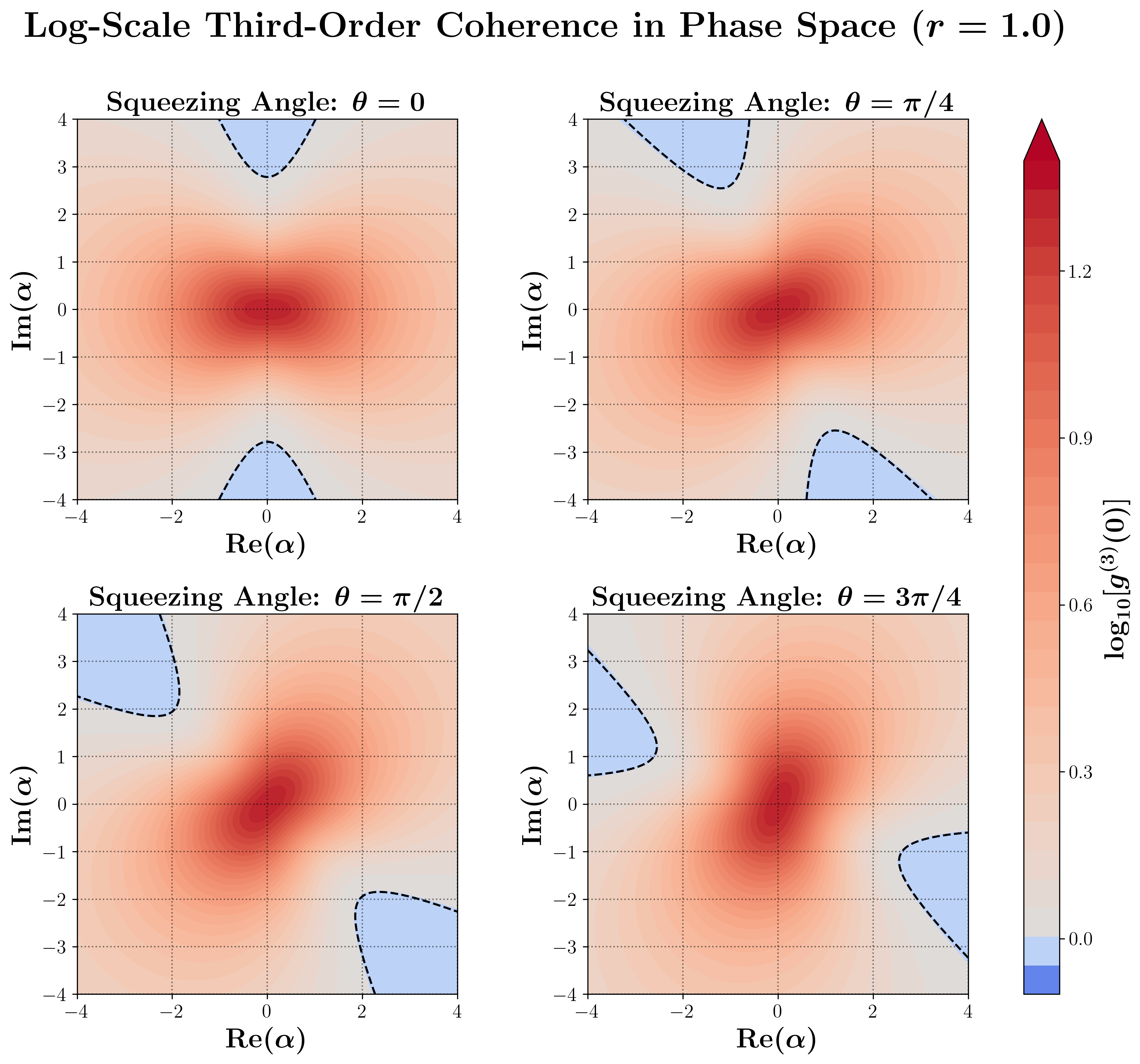}
  \caption{Phase–space maps of \(\log_{10} g^{(3)}(0)\) at \(r=1\) for four \(\theta\). The black contour marks \(g^{(3)}(0)=1\) (\(\log_{10}=0\)).}
  \label{fig:g3_rotated}
\end{figure*}

\subsubsection*{Fourth-Order Coherence, \texorpdfstring{$g^{(4)}(0)$}{}}
For $k=4$, the factorial moment $\mathcal{N}_4$ is given by:

 \begin{align}
\mathcal{N}_4 = & |\alpha|^8 + |\alpha|^6 \Big( 16\sinh^2 r + 12\sinh r \cosh r \cos(2\varphi-\theta) \Big) \nonumber \\
&+ |\alpha|^4 \Big( 108\sinh^4 r + 36\sinh^2 r + 96\sinh^3 r\cosh r \cos(2\varphi-\theta) \nonumber \\ & \hspace{1.5cm} + 6\sinh^2 r\cosh^2 r \cos(2(2\varphi-\theta)) \Big) \nonumber \\
&+ |\alpha|^2 \Big( 240\sinh^6 r + 144\sinh^4 r + (180\sinh^5 r + 36\sinh^3 r)\cosh r \cos(2\varphi-\theta) \Big) \nonumber \\
&+ 105\sinh^8 r + 90\sinh^6 r + 9\sinh^4 r.
\end{align}

The fourth-order coherence \(g^{(4)}(0)=\mathcal N_4/\mathcal N_1^{4}\) is plotted on a logarithmic scale (matching the \(g^{(3)}(0)\) panels) to avoid color saturation and emphasize the rotational dependence. Figure~\ref{fig:g4_rotated} shows \(\log_{10} g^{(4)}(0)\) with the neutral contour \(g^{(4)}(0)=1\) (\(\log_{10}=0\)) overlaid for reference.

\begin{figure}[t]
  \centering
  \includegraphics[width=.8\columnwidth]{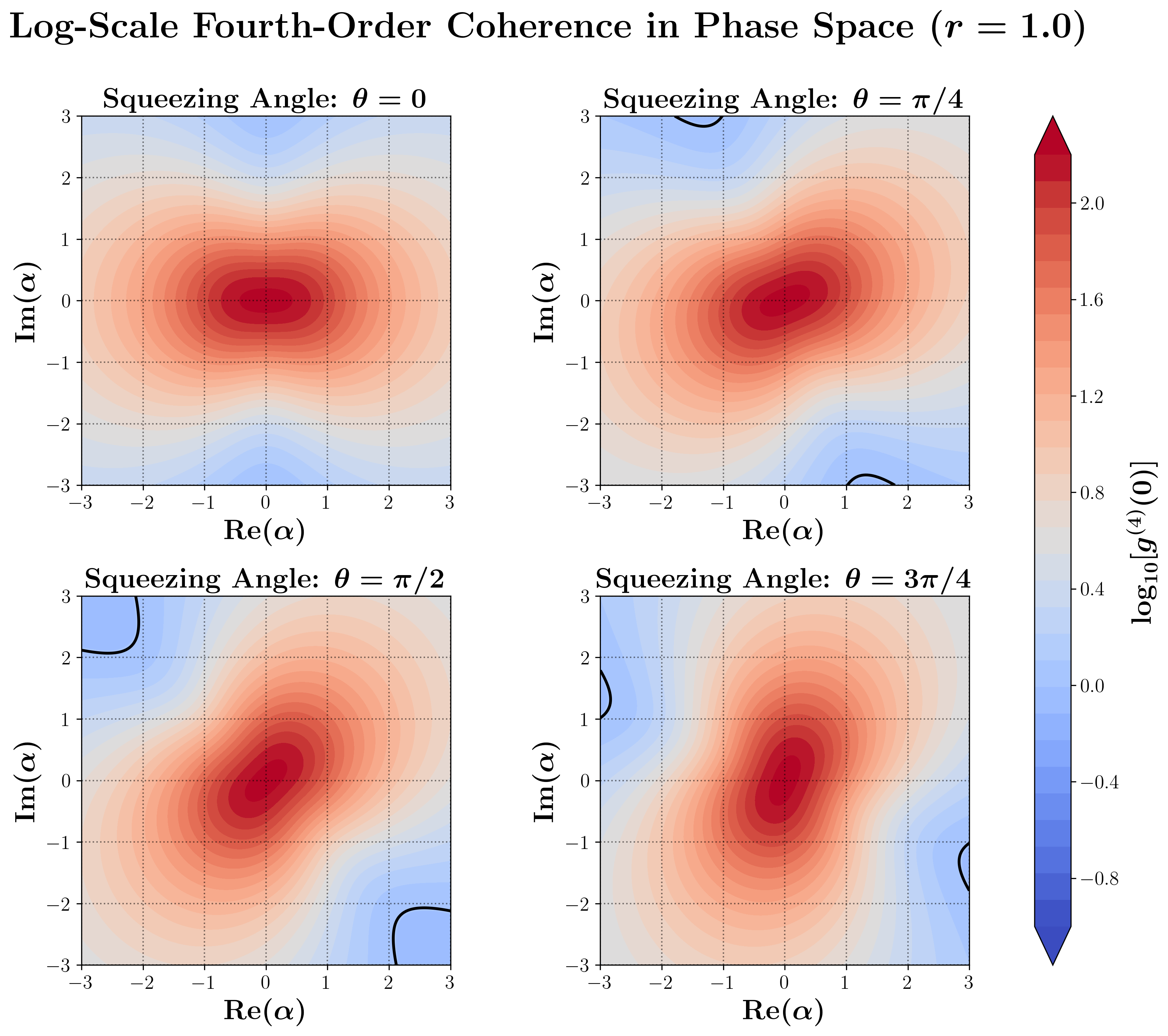}
  \caption{Phase-space maps of \(\log_{10} g^{(4)}(0)\) at \(r=1\) for four squeezing angles \(\theta\).
  The black contour marks \(g^{(4)}(0)=1\) (\(\log_{10}=0\)).
  A logarithmic scale reveals both the large central bunching (bright region) and the phase-rotated
  angular structure that would be compressed on a linear scale.}
  \label{fig:g4_rotated}
\end{figure}

\begin{figure*}[t]
  \centering
  \includegraphics[width=\textwidth]{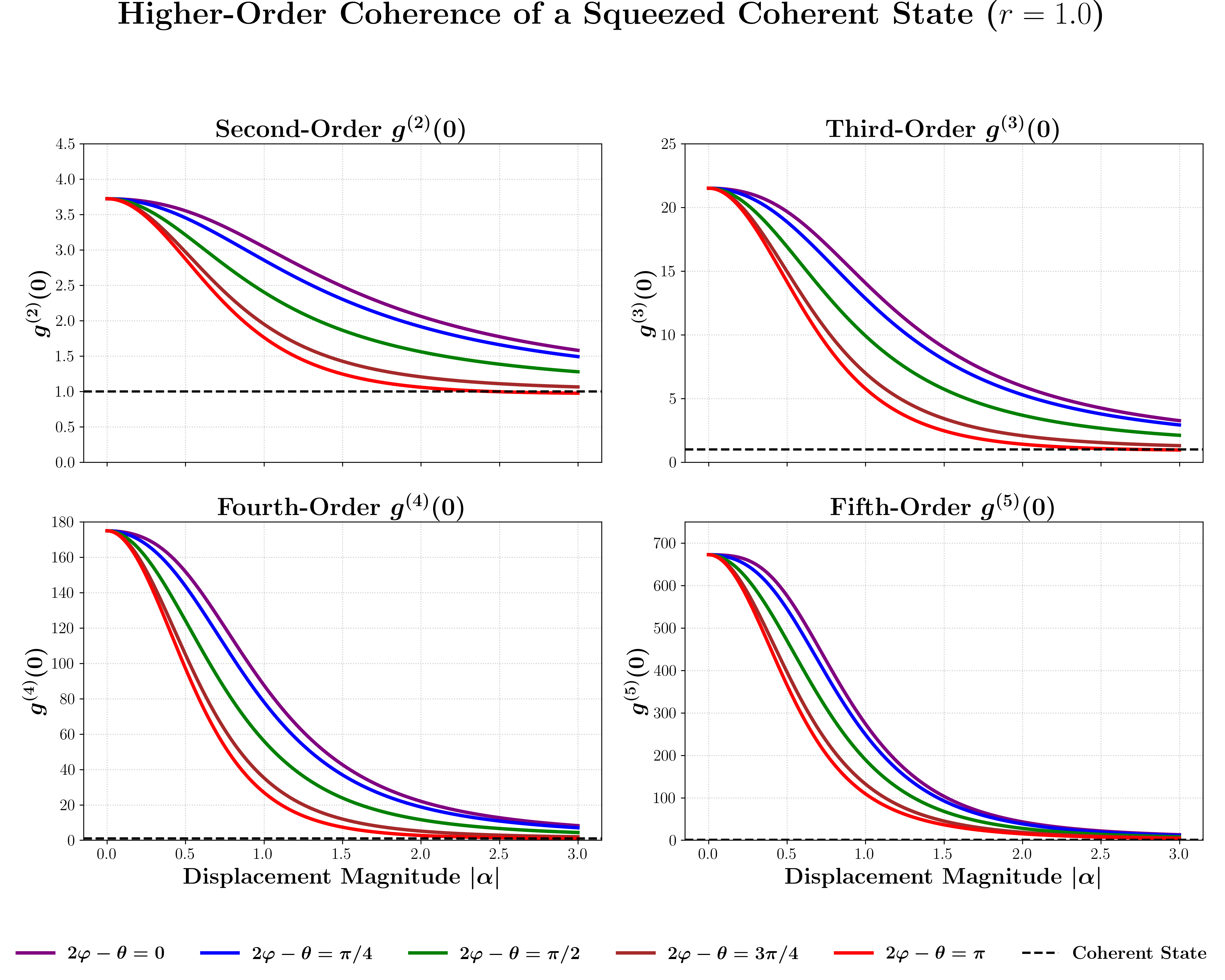}
  \caption{\textbf{Higher-order coherence of a squeezed coherent state} (\(r=1.0\)).
  Zero-delay correlations \(g^{(k)}(0)\) for \(k=2,3,4,5\) versus displacement \(|\alpha|\),
  shown for relative phases \(2\varphi-\theta\in\{0,\pi/4,\pi/2,3\pi/4,\pi\}\).
  Dashed lines mark the coherent-state benchmark \(g^{(k)}(0)=1\).
  Across all orders and phases, the curves approach unity with increasing \(|\alpha|\); for small \(|\alpha|\) the departure from unity amplifies with \(k\) and is phase-ordered, with \(2\varphi-\theta\approx\pi\) giving the strongest suppression.}
  \label{fig:gk_hierarchy_scs}
\end{figure*}

Fig.~\ref{fig:gk_hierarchy_scs} summarizes the trends of the normalized zero-delay correlations \(g^{(k)}(0)\), as obtained from Eqs.~\eqref{eq:M0}–\eqref{eq:M3} together with the definition \(g^{(k)}(0)=\mathcal{N}_k/\mathcal{N}_1^{\,k}\) (for a single squeezed coherent state, see also Eq.~\eqref{eq:g2_single_state}). The consolidated panels for \(k=2,3,4,5\) sweep the displacement magnitude \(|\alpha|\) and the relative phase \(2\varphi-\theta\). In every case the curves relax to the coherent-state benchmark \(g^{(k)}(0)=1\) as \(|\alpha|\) increases; at small \(|\alpha|\) the deviation from unity grows rapidly with \(k\) and is monotonically ordered by the phase, with the strongest suppression near \(2\varphi-\theta\approx\pi\).

\subsection{Superposition of Squeezed Coherent States}

We now use the single-state results to analyze the displaced Janus state. The $k$-th factorial moment of this superposition, $\mathcal{N}_k(\Psi)$, expands into two diagonal terms and one interference term:
\begin{align}
    \mathcal{N}_k(\Psi) = |\chi|^2 \mathcal{N}_k(\xi,\alpha) + |\eta|^2 \mathcal{N}_k(\zeta,\alpha) + 2 \mathrm{Re} \Big[ \chi \eta^* \, \mathcal{M}_k(\zeta,\xi,\alpha) \Big].
    \label{eq:Nk_superposition_expanded}
\end{align}
The diagonal moments $\mathcal{N}_k(\xi,\alpha)$ and $\mathcal{N}_k(\zeta,\alpha)$ are precisely the single-state quantities analyzed in the previous subsection. The interference is governed by the off-diagonal matrix element $\mathcal{M}_k(\zeta,\xi,\alpha) \equiv \langle \alpha , \zeta| a^{\dagger k} a^k \ket{\alpha, \xi}$, for which we have derived general analytical expressions.

Combining these components, we arrive at the general-order correlation function for the displaced Janus state:
\begin{align}
g^{(k)}(0) = \frac{|\chi|^2 \mathcal{N}_k(\xi,\alpha) + |\eta|^2 \mathcal{N}_k(\zeta,\alpha) + 2 \mathrm{Re} \Big[ \chi \eta^* \, \mathcal{M}_k(\zeta,\xi,\alpha) \Big]}{\Big(|\chi|^2 \mathcal{N}_1(\xi,\alpha) + |\eta|^2 \mathcal{N}_1(\zeta,\alpha) + 2 \mathrm{Re} \Big[ \chi \eta^* \, \mathcal{M}_1(\zeta,\xi,\alpha) \Big]\Big)^k}.
\label{eq:gk_Janus_final}
\end{align}
This final expression encapsulates the complete photon statistics of the state. It demonstrates how nonclassical properties are determined by a complex interplay between the characteristics of the individual states $\mathcal{N}_k$ and the quantum interference between them, captured by the off-diagonal term $\mathcal{M}_k$.

\subsection*{Normalization Condition}

For the state $|\Psi\rangle = \chi |\alpha, \xi \rangle + \eta |\alpha, \zeta \rangle$ to be physically valid, its norm must be one, i.e., $\langle\Psi|\Psi\rangle = 1$. Expanding this inner product yields the general constraint on the complex coefficients $\chi$ and $\eta$:
\begin{align}
    |\chi|^2 + |\eta|^2 + 2 \mathrm{Re} \Big[ \eta^*\chi \langle \alpha, \zeta\ket{\alpha, \xi} \Big] = 1.
    \label{eq:norm_expanded}
\end{align}
The crucial interference term depends on the state overlap $\langle \alpha, \zeta\ket{\alpha, \xi}$. This overlap is equivalent to the zeroth-order matrix element, $\mathcal{M}_0$, which is given by Eq.~\eqref{eq:M0}. Substituting this explicit formula for the overlap into Eq.~\eqref{eq:norm_expanded} gives the complete and final normalization condition for the displaced Janus state.
\begin{align}
    |\chi|^2 + |\eta|^2 + 2 \mathrm{Re} \Bigg[ \frac{\eta^*\chi}{\sqrt{\cosh s \cosh r - \sinh s \sinh r \, e^{i(\theta - \phi)}}} \Bigg] = 1.
\end{align}

\subsection*{Explicit Results for Low-Order Correlations}

While the general formula is compact, the explicit expressions reveal the full complexity and richness of the interference.
\subsubsection*{Second-Order Coherence, \texorpdfstring{$g^{(2)}(0)$}{}}
The second-order coherence is given by $g^{(2)}(0) = \mathcal{N}_2(\Psi) / \mathcal{N}_1(\Psi)^2$. The numerator, the second factorial moment of the superposition, is:
\begin{align}
\mathcal{N}_2(\Psi) = &|\chi|^2 \Big[ |\alpha|^4 + |\alpha|^2 \left( 4\sinh^2 r + 2\sinh r \cosh r \cos(2\varphi - \theta) \right) + 3\sinh^4 r + \sinh^2 r \Big] \nonumber \\
&+ |\eta|^2 \Big[ |\alpha|^4 + |\alpha|^2 \left( 4\sinh^2 s + 2\sinh s \cosh s \cos(2\varphi - \phi) \right) + 3\sinh^4 s + \sinh^2 s \Big] \nonumber \\
&+ 2 \mathrm{Re} \Bigg\{ \frac{\chi \eta^*}{\sqrt{\cosh r \cosh s}\,(1-z)^{5/2}} \times \bigg[ (1-z)^2 |\alpha|^4  + (2z^2+z) \nonumber \\
& 
\phantom{+ 2 \mathrm{Re} \Bigg\{ }
+ (1-z)|\alpha|^2 \Big( 4z + \tanh s\, e^{i(2\varphi - \phi)} + \tanh r\, e^{i(\theta - 2\varphi)} \Big) \bigg] \Bigg\}.
\end{align}
The denominator is constructed from the mean photon number, $\mathcal{N}_1(\Psi)$, which is given by:
\begin{align}
\mathcal{N}_1(\Psi) = &|\chi|^2 \left( |\alpha|^2 + \sinh^2 r \right) + |\eta|^2 \left( |\alpha|^2 + \sinh^2 s \right) \nonumber \\
&+ 2 \mathrm{Re} \left[ \frac{\chi \eta^*}{\sqrt{\cosh r \cosh s}\,(1-z)^{3/2}} \Big( |\alpha|^2(1-z) + z \Big) \right].
\end{align}
\subsubsection*{Third-Order Coherence, \texorpdfstring{$g^{(3)}(0)$}{}}

The third-order coherence, $g^{(3)}(0) = \mathcal{N}_3(\Psi) / \mathcal{N}_1(\Psi)^3$, is built from the same denominator. The numerator, $\mathcal{N}_3(\Psi)$, is constructed from the single-state moments $\mathcal{N}_3(\xi,\alpha)$ and $\mathcal{N}_3(\zeta,\alpha)$, and the off-diagonal matrix element $\mathcal{M}_3(\zeta,\xi,\alpha)$ from Eq.~\eqref{eq:M3}:
\begin{align}
\mathcal{N}_3(\Psi) = &|\chi|^2 \Big[ \mathcal{N}_3(\xi,\alpha) \Big] + |\eta|^2 \Big[ \mathcal{N}_3(\zeta,\alpha) \Big] + 2 \mathrm{Re} \Big[ \chi \eta^* \, \mathcal{M}_3(\zeta,\xi,\alpha) \Big].
\end{align}
Here we use the symbolic forms for the components of $\mathcal{N}_3(\Psi)$ as writing the fully expanded expression becomes exceptionally long. However, it can be constructed straightforwardly by substituting the previously derived formulas.

\subsection{The \texorpdfstring{$\alpha=0$}{alpha=0} Limit: Recovering the Janus State}

Setting the common displacement to zero in our displaced Janus formulas reproduces the original Janus state— a superposition of squeezed \emph{vacua}—with all overlaps, cross terms, and moments reducing accordingly. Figure~\ref{fig:wigner-janus} illustrates this limit: at \(\alpha=0\) the components share a common mean, and the covariance mismatch alone generates the characteristic interference fringes and negativity.

\begin{figure*}[t]
  \centering
  \includegraphics[width=.8\textwidth]{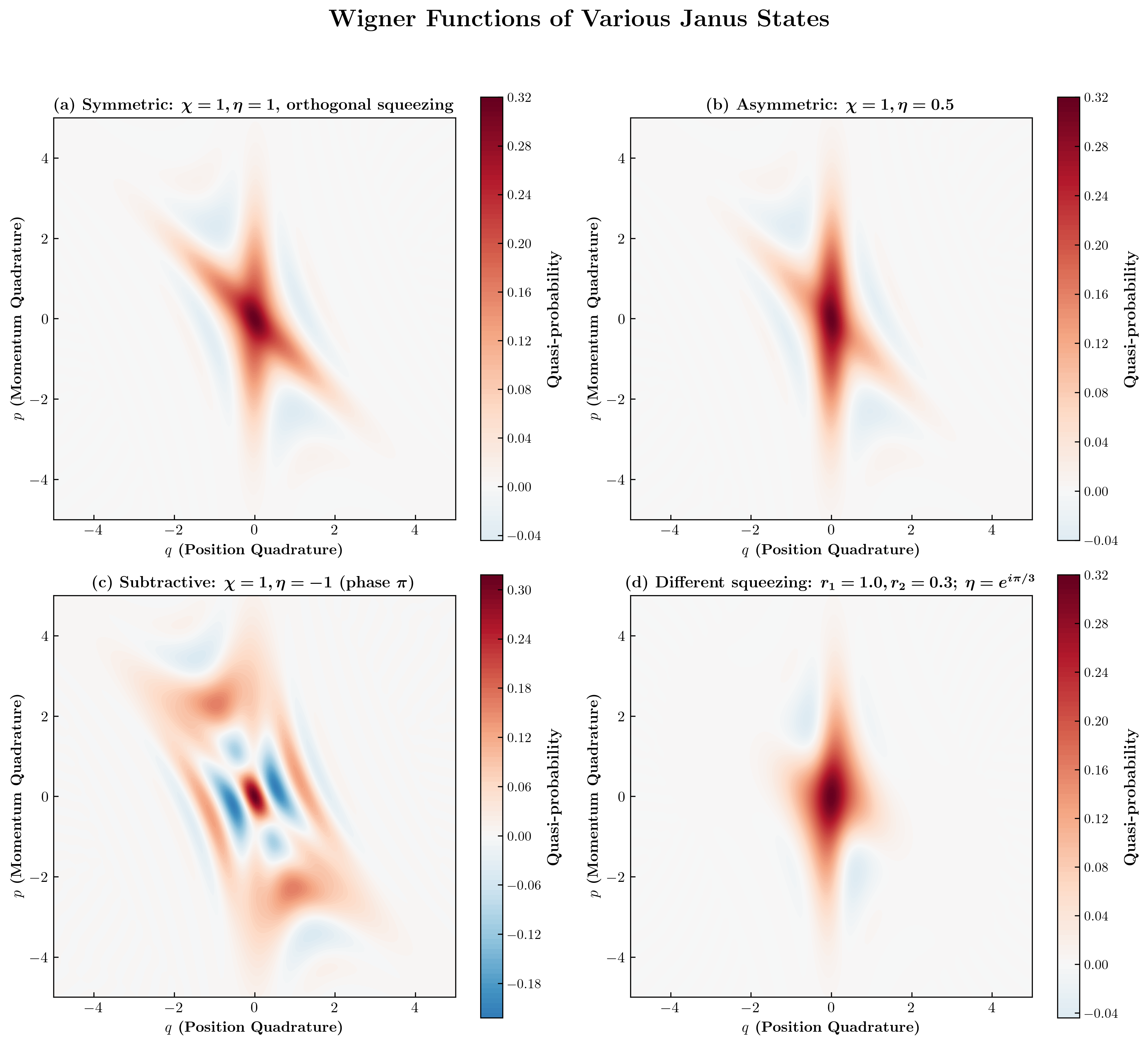}
  \caption{Janus (\(\alpha=0\)) Wigner functions for representative parameter choices, showing covariance-mismatch-driven negativity and morphology changes with relative squeezing.}
  \label{fig:wigner-janus}
\end{figure*}

First, the state vector itself simplifies as expected. A squeezed coherent state with zero displacement is, by definition, a squeezed vacuum state:
\begin{align}
    |\alpha=0, \xi\rangle = D(0)S(\xi)|0\rangle = S(\xi)|0\rangle \equiv |\xi\rangle.
\end{align}
Therefore, the superposition becomes a coherent superposition of two distinct squeezed vacua:
\begin{align}
    |\Psi\rangle\Big|_{\alpha=0} = \chi |\xi\rangle + \eta |\zeta\rangle.
\end{align}
Consequently, the photon statistics must also reduce to those of this undisplaced state. This is confirmed by setting $|\alpha|=0$ in our general formulas for the moments of the superposition state. The single-state moment, $\mathcal{N}_k(\xi,\alpha)$, correctly simplifies to that of a squeezed vacuum, $\mathcal{N}_k^{(\text{vac})}(r)$. The off-diagonal matrix element, $\mathcal{M}_k(\zeta,\xi,\alpha)$, becomes the interference term between two squeezed vacuum states:
\begin{align}
    \mathcal{M}_k(\zeta,\xi,\alpha)\Big|_{\alpha=0} = \langle \zeta | a^{\dagger k} a^k | \xi \rangle.
\end{align}
As a concrete example, setting $|\alpha|=0$ in the general expressions for $\mathcal{N}_1(\Psi)$ and $\mathcal{N}_2(\Psi)$ yields the second-order coherence function for the undisplaced Janus state:
\begin{align}
g^{(2)}(0)\Big|_{\alpha=0} = \frac{|\chi|^2 \mathcal{N}_2^{(\text{vac})}(r) + |\eta|^2 \mathcal{N}_2^{(\text{vac})}(s) + 2 \mathrm{Re} \left[ \chi \eta^* \mathcal{M}_2(\zeta,\xi,0) \right]}{\left( |\chi|^2 \sinh^2 r + |\eta|^2 \sinh^2 s + 2 \mathrm{Re} \left[ \chi \eta^* \mathcal{M}_1(\zeta,\xi,0) \right] \right)^2},
\end{align}
where $\mathcal{N}_2^{(\text{vac})}(r) = 3\sinh^4 r + \sinh^2 r$, and the explicit formulas for the interference terms $\mathcal{M}_1(\zeta,\xi,0)$ and $\mathcal{M}_2(\zeta,\xi,0)$ were derived previously. This confirms that our framework provides a true generalization, with the undisplaced case naturally embedded within the more comprehensive model.

\subsubsection*{Special Case: Antisymmetric Superposition with Vanishing Squeezing}

A particularly interesting limit occurs for an "antisymmetric" superposition where $\eta = -\chi$, with equal squeezing magnitudes ($r=s$) and opposing squeezing axes ($\theta = \phi+\pi$). In the limit of small squeezing ($r \to 0$), the state normalization condition imposes the constraint $2|\chi|^2 r^2 \approx 1$. This unique interplay produces a state with strong quantum correlations. A detailed symbolic calculation of the series expansion in the squeezing parameter $r$ reveals that the first-order correction term (of order $r^2$) vanishes for both $g^{(2)}(0)$ and $g^{(3)}(0)$. This indicates that the leading-order photon statistics of this state are remarkably robust against small amounts of squeezing.

The leading-order behavior and the first non-vanishing correction (of order $r^4$) are given by the following expressions:
\begin{align}
    g^{(2)}(0) &\approx \frac{|\alpha|^4 + 8|\alpha|^2 + 2}{\left( |\alpha|^2 + 2 \right)^2} + \frac{5(2|\alpha|^4 - |\alpha|^2 + 10)}{2(|\alpha|^2+2)^3} r^4 \\
    g^{(3)}(0) &\approx \frac{|\alpha|^6 + 18|\alpha|^4 + 18|\alpha|^2}{\left( |\alpha|^2 + 2 \right)^3} + \frac{15(2|\alpha|^6 + 9|\alpha|^4 + 34|\alpha|^2 + 20)}{2(|\alpha|^2+2)^4} r^4
\end{align}
For zero displacement ($\alpha=0$), the state approaches the pure two-photon Fock state $|2\rangle$. The formulas above correctly reproduce the characteristic values $g^{(2)}(0)=1/2$ and $g^{(3)}(0)=0$ in the $r \to 0$ limit~\cite{Azizi2025Janus, Azizi2025Janus_higher}.  The $r^4$ terms describe the first deviations from the pure two-photon character as a small amount of squeezing is introduced.

In the antisymmetric, equal-amplitude Janus configuration (same displacement on both components, \(r=s\), \(\eta=-\chi\), and \(\theta=\varphi+\pi\)), a common displacement invariably increases the second-order coherence. In the small-squeezing regime, where the leading \(O(r^{2})\) terms cancel, we obtain
\begin{align}
g^{(2)}(0)\approx \frac{|\alpha|^{4}+8|\alpha|^{2}+2}{\big(|\alpha|^{2}+2\big)^{2}}.
\end{align}
Hence,
\begin{align}
g^{(2)}(0)-\tfrac{1}{2}=\frac{|\alpha|^{2}\big(|\alpha|^{2}+12\big)}{2\big(|\alpha|^{2}+2\big)^{2}}\ge 0,
\end{align}
with equality only at \(|\alpha|=0\). Thus the displaced Janus never beats the undisplaced one: the global minimum within this family is achieved at zero displacement, where the state approaches \(\ket{2}\) and attains the floor \(g^{(2)}(0)=\tfrac{1}{2}\). Physically, adding a common displacement breaks the even-photon parity structure that underpins the antibunching of the \(\alpha=0\) superposition, thereby raising \(g^{(2)}(0)\).

\subsection*{Understanding the Small-Squeezing Expansion: Role of Interference and Optimization}

To clarify the difference in the small-squeezing expansions between the displaced and undisplaced Janus states, we examine the origin of the $r^2$ correction term in $g^{(2)}(0)$. In the antisymmetric superposition with equal amplitudes ($|\chi| = |\eta|$, $\eta = -\chi$, $\theta = \phi + \pi$), the $r^2$ term cancels due to destructive interference, regardless of displacement $\alpha$. However, in the optimized undisplaced case (where amplitudes are unequal to minimize $g^{(2)}(r)$ for finite $r$), this symmetry is broken, introducing a non-zero $r^2$ term.

\subsubsection*{\texorpdfstring{Equal-Amplitude Case: Cancellation of the $r^2$ Term}{}}
In this configuration, the normalization requires $|\chi|^2 \propto 1/r^2$ as $r \to 0$, preserving the limit to a (displaced) Fock state $|2\rangle$. The factorial moments expand as:
\begin{align*}
\mathcal{N}_1(\Psi) &\approx 2 + O(r^4) \quad (\alpha = 0), \quad \mathcal{N}_1(\Psi) \approx |\alpha|^2 + 2 + O(r^4) \quad (\alpha \neq 0), \\
\mathcal{N}_2(\Psi) &\approx 2 + O(r^4) \quad (\alpha = 0), \quad \mathcal{N}_2(\Psi) \approx |\alpha|^4 + 8|\alpha|^2 + 2 + O(r^4) \quad (\alpha \neq 0).
\end{align*}
The $r^2$ contributions from diagonal terms ($\mathcal{N}_k(\xi, \alpha) + \mathcal{N}_k(\zeta, \alpha)$) are exactly opposed by the interference ($-2 \mathrm{Re}[\mathcal{M}_k]$), leading to $g^{(2)}(0) \approx \frac{1}{2} + 0 \cdot r^2 + \frac{25}{8} r^4 + O(r^6)$ for $\alpha = 0$, or the analogous displaced form.

\begin{figure*}[b]
  \centering
  \includegraphics[width=\textwidth]{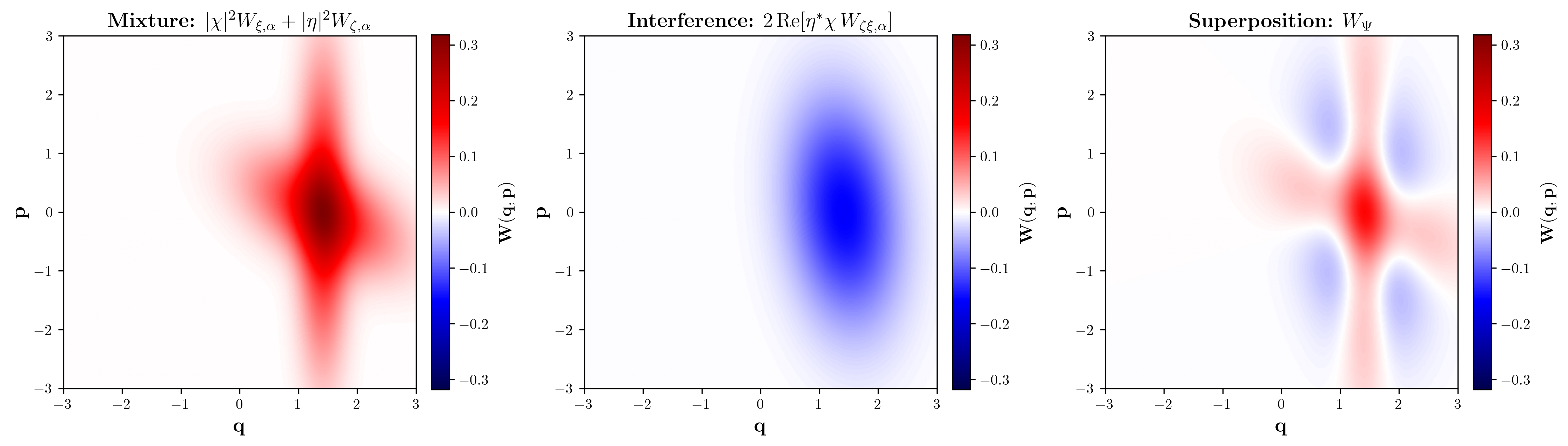}
  \caption{Wigner-function anatomy of a coherent superposition: (left) weighted mixture of the two Gaussians; (middle) interference term; (right) full superposition \(W_\Psi\).}
  \label{fig:wigner-decomp}
\end{figure*}

For the undisplaced ($\alpha = 0$) Janus state optimized to minimize $g^{(2)}(r)$, the amplitudes $|\chi| \neq |\eta|$ break the interference symmetry. This yields the exact expression:
\begin{align}
g^{(2)}(r) = \frac{12 \sinh^{10} r + 40 \sinh^{8} r + 51 \sinh^{6} r + 28 \sinh^{4} r + 11 \sinh^{2} r + 2}{4 \sinh^{10} r + 16 \sinh^{8} r + 29 \sinh^{6} r + 29 \sinh^{4} r + 16 \sinh^{2} r + 4},
\end{align}
with expansion $g^{(2)}(r) \approx \frac{1}{2} + \frac{3}{4} \sinh^2 r + \frac{3}{8} \sinh^4 r + O(\sinh^6 r)$. The optimization tunes the ratio $|\chi|/|\eta|$ for each $r$, preventing full cancellation and introducing the positive $\frac{3}{4} \sinh^2 r \approx \frac{3}{4} r^2$ term.

This highlights how quantum interference and coefficient choice control the nonclassical photon statistics in the small-squeezing regime.

\end{widetext}

\section{The Displaced Janus State as a Non-Gaussian Resource} \label{sec:Non-Gaussian}

A central challenge in quantum computing is the identification and characterization of resources that enable a quantum advantage. In the continuous-variable (CV) paradigm, the Gottesman--Knill theorem dictates that a quantum device restricted to Gaussian states, Gaussian operations, and Gaussian measurements can be efficiently simulated by a classical computer. To achieve universal quantum computation, at least one non-Gaussian element is required. States with negative Wigner functions are a particularly powerful class of non-Gaussian resources that can unlock this quantum advantage. The displaced Janus (DJ) state is a prime example, whose utility can be characterized both structurally via its Wigner function and operationally via the quantum Fisher information (QFI).

\subsection{Wigner Function}

We begin with the single squeezed coherent state $\ket{\alpha,\xi}$, where $\xi=r\,e^{i\theta}$ and $\alpha=(\langle q\rangle+i\langle p\rangle)/\sqrt{2}$. Its covariance is
\begin{align}
  V_\xi=&\frac{1}{2}\,R\!\left(\frac{\theta}{2}\right)
  \begin{pmatrix} e^{-2r} & 0 \\[2pt] 0 & e^{2r} \end{pmatrix}
  R\!\left(\frac{\theta}{2}\right)^{\!T},
  \nn\\
  R(\phi)=&\begin{pmatrix}\cos\phi & -\sin\phi\\ \sin\phi & \cos\phi\end{pmatrix}.
\end{align}
With $\mathbf{x}=(q,p)^T$ and $\mathbf{x}_\alpha=(\sqrt{2}\,\mathrm{Re}\,\alpha,\sqrt{2}\,\mathrm{Im}\,\alpha)^T$, the Wigner function is the Gaussian
\begin{align}
  W_{\xi,\alpha}(q,p)
  =\frac{1}{2\pi\sqrt{\det V_\xi}}\,
  \exp\!\left[-\frac{1}{2}\,(\mathbf{x}-\mathbf{x}_\alpha)^{T}V_\xi^{-1}(\mathbf{x}-\mathbf{x}_\alpha)\right],
\end{align}
which for a pure squeezed state ($\det V_\xi=\tfrac14$) simplifies to $W_{\xi,\alpha}=\pi^{-1}\exp[\cdots]$. Thus any single Gaussian $W_{\xi,\alpha}$ is everywhere non-negative.

\begin{figure*}
  \centering
  \includegraphics[width=.8\textwidth]{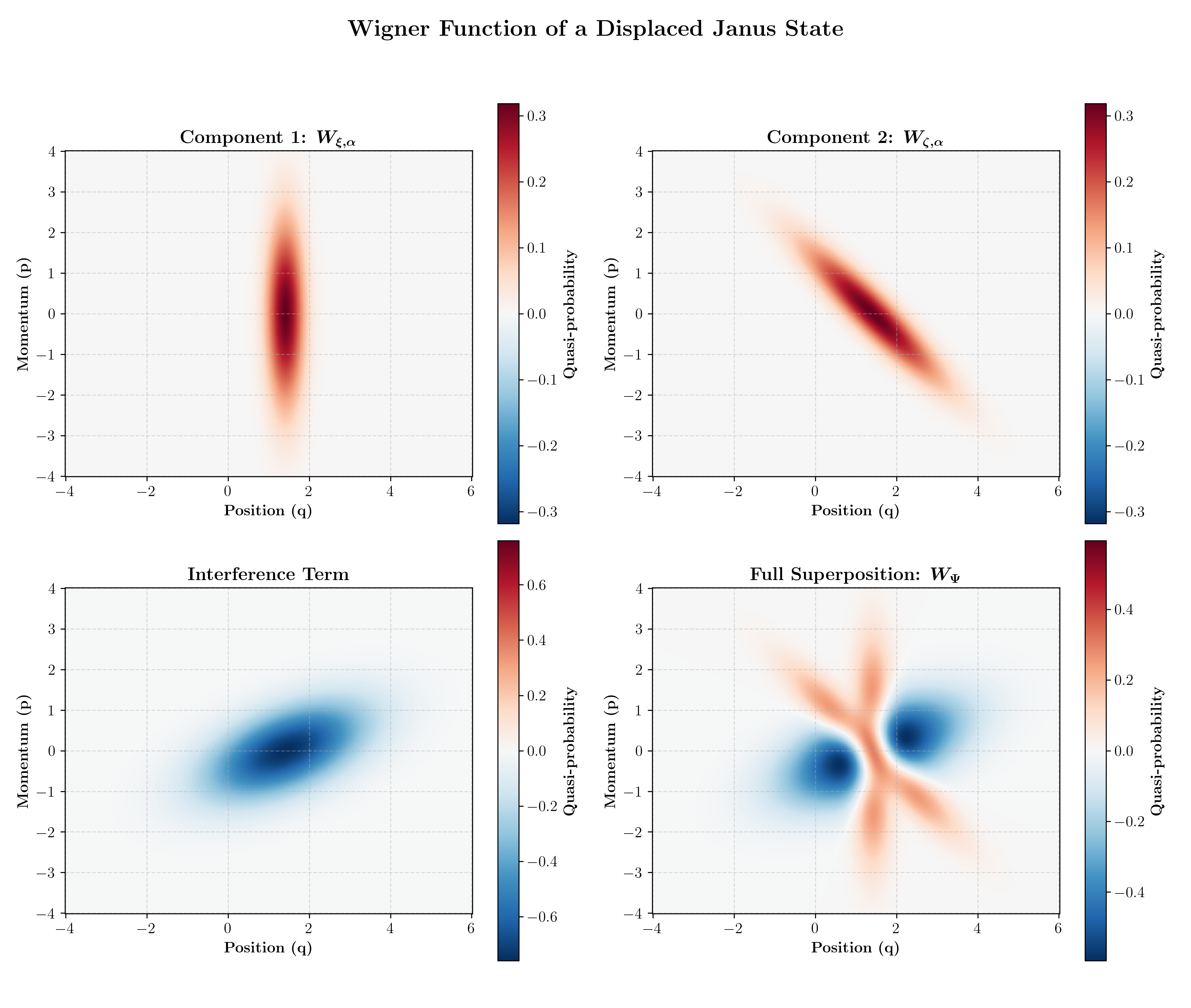}
  \caption{Displaced-Janus anatomy: component \(W_{\xi,\alpha}\), component \(W_{\zeta,\alpha}\), interference term, and full superposition \(W_\Psi\).}
  \label{fig:wigner-dj-anatomy}
\end{figure*}

\medskip\noindent\textit{Displaced Janus superposition and decomposition.}
Consider the displaced Janus (DJ) state
\begin{align}
  \ket{\Psi}=\chi\,\ket{\alpha,\xi}+\eta\,\ket{\alpha,\zeta},\qquad
  \xi=r\,e^{i\theta},\;\; \zeta=s\,e^{i\varphi},
\end{align}
whose two components share the same mean $\alpha$ but have distinct covariances $V_\xi$ and $V_\zeta$. Writing $\beta=(q+ip)/\sqrt{2}$, the Wigner function splits into a classical mixture plus a purely quantum interference term,
\begin{align}
    W_{\Psi}(\beta)
    =& |\chi|^{2}\,W_{\xi,\alpha}(\beta)
     + |\eta|^{2}\,W_{\zeta,\alpha}(\beta) \nn\\
     &
     + 2\,\mathrm{Re}\!\left[\eta^{*}\chi\,W_{\zeta\xi,\alpha}(\beta)\right].
    \label{eq:Wigner_superposition}
\end{align}
The mixture–plus–interference decomposition of the DJ Wigner function is illustrated in Fig.~\ref{fig:wigner-decomp}.

A normalization-consistent closed form for the cross term (with $dq\,dp=2\,d^2\beta$) is
\begin{align}
    W_{\zeta\xi,\alpha}(\beta)
    =& \frac{\langle \zeta|\xi\rangle}{2\pi}\,\sqrt{\det \Sigma^{-1}}\,
    \exp\!\Big[
      -A\,|\beta-\alpha|^{2} \nn\\
      &
      + B\,(\beta-\alpha)^{2}
      + B^{*}\,(\beta^{*}-\alpha^{*})^{2}
    \Big],
    \label{eq:cross_Wigner_complex}
\end{align}
where $\Sigma=\tfrac{1}{2}(V_\zeta+V_\xi)$ and
\begin{align}
    A=\tfrac{1}{2}\!\left(\Sigma^{-1}_{11}+\Sigma^{-1}_{22}\right),
    \qquad
    B=\tfrac{1}{4}\!\left(\Sigma^{-1}_{11}-\Sigma^{-1}_{22}-2i\,\Sigma^{-1}_{12}\right).
\end{align}
This expression integrates to $\langle\zeta|\xi\rangle$, reduces to the usual single-state Gaussian when $\zeta=\xi$, and shows that a nonzero $\mathrm{Im}\,B$ (mismatched squeezing axes) generates the phase-sensitive fringes. A full derivation and normalization audit are provided in App.~\ref{app:WignerDerivations}.

\noindent\textit{Phase-space portraits, control knobs, and negativity.}
Figs.~\ref{fig:wigner-dj-anatomy} (anatomy) and \ref{fig:wigner-dj-gallery} (gallery) visualize Eq.~\eqref{eq:Wigner_superposition} for the displaced Janus (DJ) state. The \emph{anatomy} figure decomposes the DJ Wigner into its two positive Gaussian components and the oscillatory interference term, highlighting that—because both components share the same mean \(\alpha\)—nonclassicality stems from a \emph{covariance mismatch} \(V_\zeta\neq V_\xi\), not mean separation. The \emph{gallery} sweeps the control knobs to show how the pattern responds: (i) the relative squeezing phase \(\Delta\theta=\theta-\varphi\) and the complex coefficient \(B\) in \eqref{eq:cross_Wigner_complex} set fringe orientation/curvature; (ii) the weights \((\chi,\eta)\) set interference contrast; and (iii) \(|\alpha|\) repositions negativity lobes away from the origin. Negativity appears wherever the interference locally exceeds the classical mixture of the two Gaussians.

\medskip\noindent\textit{Fringe orientation and spacing.}
Let $\delta\beta=\beta-\alpha=\rho\,e^{i\varphi}$ and write $B=|B|e^{i\arg B}$. The oscillatory part of the exponent is
\begin{align}
  B\,\delta\beta^{2}+B^{*}\,\delta\beta^{*2}
  = 2|B|\,\rho^{2}\cos\!\big(2\varphi+\arg B\big),
\end{align}
so the fringes are quadratic (hyperbolic in $(q,p)$), rotated by $\tfrac{1}{2}\arg B$, and become denser as $|B|$ increases (larger mismatch or $\Delta\theta\!\to\!\pm\pi/2$). Modest $|\alpha|$ can move negativity lobes away from the origin to improve tomographic SNR; the subtractive choice $\eta=-\chi$ maximizes contrast, whereas asymmetric weights ($|\eta/\chi|\ll1$) suppress it.

\begin{figure*}[t]
  \centering
  \includegraphics[width=.8\textwidth]{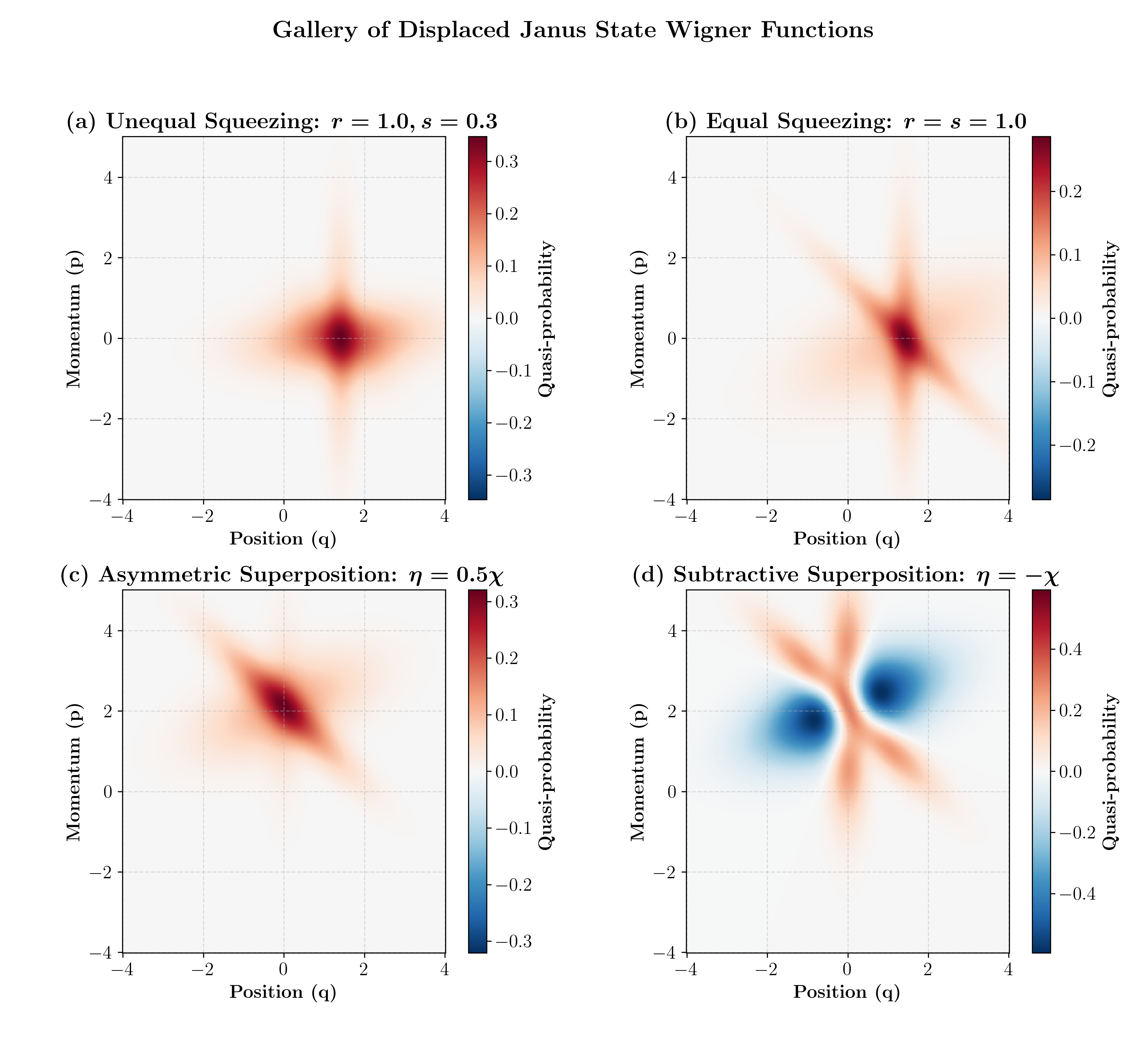}
  \caption{Gallery of displaced-Janus Wigner functions across parameters \((r,s,\theta-\phi,\chi,\eta,\alpha)\), illustrating how covariance mismatch and weights sculpt negativity and fringe geometry.}
  \label{fig:wigner-dj-gallery}
\end{figure*}

\begin{remark}[Covariance-driven negativity at $|\alpha|=0$]
When $\alpha=0$, $V_\zeta\neq V_\xi$ and $\Delta\theta\not\equiv 0$ make $B$ generically complex, producing oscillatory cross terms. By tuning $\arg(\eta^{*}\chi)$ one enforces local destructive interference and achieves $W_\Psi<0$. This has no analog in cat/squeezed-cat states where $V_\zeta=V_\xi$, explaining the DJ state’s negativity without mean separation (cf.\ Table~\ref{tab:state_comparison}).
\end{remark}

\subsection{Quantum Fisher Information and Metrological Applications}

The displaced Janus state, with its rich non-Gaussian structure, offers significant potential for quantum metrology. Suppose a parameter $\lambda$ is imprinted on the probe $|\Psi(\lambda)\rangle$. The ultimate precision obeys the quantum Cram\'er--Rao bound $\Delta\lambda \ge 1/\sqrt{\nu F_Q(\lambda)}$, where $\nu$ is the number of trials and $F_Q(\lambda)$ is the quantum Fisher information (QFI). For a pure family,
\begin{align}
F_Q(\lambda)=4\!\left(\langle \partial_\lambda \Psi | \partial_\lambda \Psi\rangle-|\langle \Psi|\partial_\lambda \Psi\rangle|^2\right).
\end{align}
Heisenberg scaling corresponds to $F_Q\sim \bar{n}^2$ (mean photon number $\bar n$), beyond the standard quantum limit (SQL) $F_Q\sim \bar{n}$.

\subsection*{Displacement phase sensing}
Estimate the phase $\varphi$ of the displacement, $\alpha=|\alpha|e^{i\varphi}$. Since $D(|\alpha|e^{i\varphi})=e^{i\varphi \hat n}D(|\alpha|)e^{-i\varphi \hat n}$, the parameter is encoded by the linear generator $\hat n$, and for a pure probe $F_Q(\varphi)=4\,\mathrm{Var}(\hat n)$. In the antisymmetric small-squeezing limit, the state approaches $D(\alpha)|2\rangle$ with leading-order moments
\[
\mathcal{N}_1\approx |\alpha|^2+2,\qquad
\mathcal{N}_2\approx |\alpha|^4+8|\alpha|^2+2,
\]
so using $\mathrm{Var}(\hat n)=\mathcal{N}_2+\mathcal{N}_1-\mathcal{N}_1^2$ one obtains
\begin{align}
F_Q(\varphi)=4\,\mathrm{Var}(\hat n)=20\,|\alpha|^2,
\end{align}
i.e.\ $F_Q(\varphi)\propto \bar n$ with $\bar n\approx |\alpha|^2+2$: displacement-phase sensing reaches the SQL.

\subsection*{Squeezing-angle sensing}
Now estimate a squeezing angle (e.g.\ $\theta_\xi$) in the \emph{undisplaced} antisymmetric state with $r=s\to 0$ and opposite axes ($\theta_\zeta=\theta_\xi+\pi$). Writing
\begin{align}
    |\xi\rangle=&\sum_{m\ge 0} c_m\,e^{i m \theta_\xi}\,r^{m}\,|2m\rangle+O(r^{m+2}),\nn\\
|\zeta\rangle=&\sum_{m\ge 0} c_m\,e^{i m (\theta_\xi+\pi)}\,r^{m}\,|2m\rangle+ O(r^{m+2}),
\end{align}
the difference $|\xi\rangle-|\zeta\rangle$ keeps only odd $m$ (vacuum cancels, $|2\rangle$ adds, $|4\rangle$ cancels, etc.). After normalization, up to $O(r^2)$,
\begin{align}
|\Psi\rangle \;\propto\;& |\xi\rangle-|\zeta\rangle
\;\Rightarrow\; \nn\\
|\Psi\rangle =& e^{i\theta_\xi}\!\left(\,|2\rangle \;+\; 
\varepsilon\,e^{i2\theta_\xi}|6\rangle \;+\; O(r^4)\right),
\end{align}
where 
\begin{align}
\varepsilon = \frac{c_3}{c_1}\,r^2 + O(r^4),
\end{align}
and $c_m=\sqrt{(2m)!}/(2^{m} m!)$. The global phase $e^{i\theta_\xi}$ is irrelevant; the \emph{first} $\theta_\xi$-dependent \emph{relative} phase arises at order $r^2$ in the $|6\rangle$ component. Using the pure-state formula,
\begin{align}
\frac{\partial}{\partial \theta_\xi}|\Psi\rangle
=& i\,2\,\varepsilon\,e^{i2\theta_\xi}|6\rangle + O(r^4),
\\
F_Q(\theta_\xi)
=& 4\!\left(\langle \partial_{\theta_\xi}\Psi|\partial_{\theta_\xi}\Psi\rangle-|\langle \Psi|\partial_{\theta_\xi}\Psi\rangle|^2\right)
= 16\,|\varepsilon|^2 + O(r^6). \nn
\end{align}
Therefore,
\begin{align}
F_Q(\theta_\xi)
= 16\,\Big|\frac{c_3}{c_1}\Big|^{\!2}\,r^4 \;+\; O(r^6)
\;\approx\; 10\,r^4 \;+\; O(r^6).
\end{align}
In particular, $F_Q(\theta_\xi)\to 0$ as $r\to 0$. Intuitively, the antisymmetric state tends to the pure Fock state $|2\rangle$, which is independent of the squeezing angle up to a global phase; sensitivity to $\theta_\xi$ enters only via higher-$m$ admixtures that scale as $r^2$.

\subsection*{Interference-enabled super-SQL scaling}

When a parameter \(\lambda\) is imprinted via \(U(\lambda)=e^{-i\lambda G}\), a pure probe satisfies
\(F_Q(\lambda)=4\,\mathrm{Var}(G)\).
For displacement-phase sensing the encoding
\(D(|\alpha|e^{i\varphi})=e^{i\varphi\hat n}D(|\alpha|)e^{-i\varphi\hat n}\)
is generated by the linear operator \(G=\hat n\), so \(F_Q=4\,\mathrm{Var}(\hat n)=\Theta(\bar n)\) (SQL).

A qualitatively different scaling arises for \emph{quadratic} Gaussian encodings, e.g.
\[
G_{\rm sq}=\tfrac12\!\left(e^{-i\theta}a^2+e^{i\theta}a^{\dagger 2}\right).
\]
For the displaced Janus (DJ) superposition \(\ket{\Psi}=\chi\ket{\alpha,\xi}+\eta\ket{\alpha,\zeta}\) with common \(\alpha\),
move the displacement onto the operators \(a\mapsto a+\alpha\). Using Eq.~\eqref{eq:Mk_final_compact},
all contributions to \(\mathrm{Var}(G_{\rm sq})\) become finite sums of off-diagonal moments
\(\mathcal M_k(\zeta,\xi,\alpha)\) that are governed by the universal series \(F_{p,q}(z)\) of Eq.~\eqref{eq:Fpq_def}
with \(z=\tanh r\,\tanh s\,e^{i(\theta-\varphi)}\) (Eq.~\eqref{z}). The large-squeezing asymptotics follow from
the singular structure \(F_{p,q}(z)=P_{p,q}(z)/(1-z)^{(p+q+1)/2}\) in Eq.~\eqref{eq:Ppq_def} as \(z\to 1\).

\emph{Notational caution (\(\Sigma\) vs.\ \(\Sigma^{-1}\)).}
Moments are obtained from the \emph{symmetric} characteristic function (CF), whose exponent depends on
\(\Sigma=\tfrac12(V_\zeta+V_\xi)\).
By contrast, the Wigner exponent uses \(\Sigma^{-1}\).
To keep the roles distinct, we use tildes for the CF/moment coefficients.

For the \emph{centered} Gaussians,
\begin{align}
    \chi_{\zeta\xi}(\lambda)=&\frac{\langle\zeta|D(\lambda)|\xi\rangle}{\langle\zeta|\xi\rangle} 
=\exp\!\left[-\tfrac12\,\lambda_R^{T}\Sigma\,\lambda_R\right],
\\
    \Sigma=&\tfrac12(V_\zeta+V_\xi),\quad
\lambda_R=(\sqrt2\,\Re\lambda,\sqrt2\,\Im\lambda)^T, \nn
\end{align}
which in complex form reads
\begin{align}
    -\tfrac12\,\lambda_R^{T}\Sigma\,\lambda_R
=& -\tilde A\,|\lambda|^2 - \tilde B\,\lambda^2 - \tilde B^{*}\,\lambda^{*2}
,\\
\tilde A=&\tfrac12(\Sigma_{11}+\Sigma_{22})b, \nn\\
\tilde B=&\tfrac14(\Sigma_{11}-\Sigma_{22}-2i\,\Sigma_{12}). \nn
\end{align}
The standard Weyl–CF derivative identities give the cross second moments
\begin{align}
\langle \zeta|a^2|\xi\rangle=&-2\tilde B^{*}\,\langle\zeta|\xi\rangle,\nn\\
\langle \zeta|a^{\dagger2}|\xi\rangle=&-2\tilde B\,\langle\zeta|\xi\rangle,\nn\\
\tfrac12\langle \zeta|\{a^\dagger,a\}|\xi\rangle=&\tilde A\,\langle\zeta|\xi\rangle,
\end{align}
and Wick/Isserlis (symmetric ordering) yields the cross fourth moment
\[
\langle \zeta|a^{\dagger2}a^2|\xi\rangle_{\!\rm sym}
=\big(2\tilde A^{\,2}+2|\tilde B|^{2}\big)\,\langle\zeta|\xi\rangle.
\]
These are precisely the combinations entering \(\mathrm{Var}(G_{\rm sq})\) after the displacement shift.

For large squeezing \(r_{\max}=\max\{r,s\}\) and generic axis mismatch (bounded away from \(\pi\ \mathrm{mod}\ 2\pi\)),
\(\tilde A,|\tilde B|=\Theta(e^{2r_{\max}})\) because they depend on \(\Sigma\) (not \(\Sigma^{-1}\)).
Meanwhile, the overlap (Eq.~\eqref{eq:M0}) satisfies \(|\langle\zeta|\xi\rangle|=\Theta(1)\) when \(|r-s|=O(1)\) and the
axes are not nearly opposite. Since \(\bar n\sim\sinh^2 r_{\max}=\Theta(e^{2r_{\max}})\),
the cross \emph{second} moments scale as \(\Theta(\bar n)\) and the cross \emph{fourth} moments as \(\Theta(\bar n^2)\).
Therefore
\[
F_Q(\text{squeezing parameter})=4\,\mathrm{Var}(G_{\rm sq})=\Theta(\bar n^2),
\]
i.e.\ Heisenberg scaling. By contrast, displacement-phase sensing (linear generator \(G=\hat n\)) remains SQL:
\(F_Q=\Theta(\bar n)\).

\emph{Edge cases.}
(i) The scaling requires only the \emph{combinations} \(\tilde A\) and \(|\tilde B|\) to behave as
\(\Theta(e^{2r_{\max}})\); no claim is made about each entry of \(\Sigma\) separately.
(ii) If \(\theta-\varphi\to\pi\) (mod \(2\pi\)), the overlap \(|\langle\zeta|\xi\rangle|\) (Eq.~\eqref{eq:M0}) decays,
suppressing interference; our claim assumes angles bounded away from that line.

\section{Experimental feasibility and measurement protocols} \label{sec:Exp_feas}

\subsection{Platforms and state preparation}
The displaced superposition targeted here is implementable with standard continuous–variable (CV) optics: two single–mode squeezed states from a below–threshold optical parametric oscillator/optical parametric amplifier (OPO/OPA), interfered on a phase–stable beam splitter (BS), followed by a common displacement and phase control. High–quality squeezed sources at $1064~\mathrm{nm}$ using periodically poled potassium titanyl phosphate (PPKTP) are well established; polarization squeezing of about $-3.6~\mathrm{dB}$ (anti–squeezing $\sim\!+9.4~\mathrm{dB}$) has been demonstrated \cite{Lassen2007PolSqueezing}, while single–mode quadrature squeezing up to $\sim\!15~\mathrm{dB}$ has been realized in subthreshold OPOs and $10$–$12~\mathrm{dB}$ is commonly achievable in well-engineered systems \cite{vahlbruch2016detection,vahlbruch2008observation,mehmet2011squeezed}. Integrated OPO platforms further reduce technical overheads and support \emph{Gaussian-state (covariance) tomography} on chip \cite{Park2024}.

Phase-stable combining is critical. ``Squeezed–state quantum averaging’’ provides an experimentally validated route to stabilize and coherently combine variance resources on a BS \cite{Lassen2010Averaging}. In our scheme, the \emph{relative squeezing phase} $\Delta\theta$ is the principal non–Gaussianity knob; the common coherent displacement $\alpha$ is implemented on a weakly reflecting BS using a phase-locked local oscillator (LO). Throughout we target $r\in[0.4,1.1]$ ($\sim\!3.5$–$9.5~\mathrm{dB}$) and $|\alpha|\in[0.5,2.0]$ in the signal mode—compatible with typical OPO output powers, high mode matching ($\gtrsim\!95\%$), and detector linearity. Beyond passive tolerance, CV \emph{erasure–correcting} strategies using only linear optics and Gaussian ancillae have been demonstrated to protect coherence against probabilistic photon loss \cite{Lassen2010NatPhoton}; such encodings are compatible with our resource and extend the effective loss window.

\subsection{Loss maps and robustness of non–Gaussian signatures}
Optical loss is modeled by a BS channel of transmissivity $\eta$ acting as $a\mapsto\sqrt{\eta}\,a+\sqrt{1-\eta}\,v$. At the level of quasiprobabilities, the output Wigner function is the Gaussian convolution
\begin{align}
W_{\eta}(\beta)=\frac{2}{\pi(1-\eta)}\!\int d^2\gamma~\exp\!\left[-\frac{2\,|\beta-\sqrt{\eta}\,\gamma|^2}{1-\eta}\right]\,W(\gamma),
\label{eq:WignerLossConv}
\end{align}
with $d^2\gamma$ the flat phase–space measure. We use the complex-plane normalization in which the vacuum has $W(0)=1/\pi$. For our displaced superposition, the interference term controlled by $\Delta\theta$ sets the location and extent of negativity and hence its susceptibility to convolution. We therefore use \emph{loss maps}
\[
N_{\min}(\eta;r,|\alpha|,\Delta\theta):=\min_{\beta\in\mathbb{C}} W_\eta(\beta),
\]
and also monitor the Kenfack–Życzkowski negativity metric \cite{Kenfack2004}. Representative, experimentally relevant trends are:
\begin{itemize}
\item For $r\!\approx\!0.7$ ($\sim\!6~\mathrm{dB}$), $|\alpha|\!\approx\!1.0$, $\Delta\theta\!\simeq\!\pi/2$, we find $N_{\min}\!\lesssim\!-0.05$ persisting up to $\eta\gtrsim0.6$, vanishing smoothly around $\eta\!\approx\!0.5\pm0.05$ depending on the displacement axis and reconstruction filter.
\item For $r\!\approx\!1.0$ ($\sim\!9~\mathrm{dB}$), fringes broaden and tolerate higher loss (similar negativity to $\eta\gtrsim0.7$), whereas very small $|\alpha|$ localizes negativity and makes it more fragile to blurring.
\end{itemize}
These maps convert measured insertion losses (escape efficiency, propagation, mode matching, detector quantum efficiency) into operating points $(r,|\alpha|,\Delta\theta)$ where negativity remains resolvable with available sampling. For channels with excess Gaussian noise, \emph{Gaussian adaptation} (pre/post-squeezing matched to the channel) can further protect the negative regions of $W$ \cite{Filip2013Gaussian}.

\subsection{\texorpdfstring{Measuring $W(\beta)$, $g^{(k)}(0)$, and factorial moments}{}}
\subsubsection*{Balanced homodyne tomography.}
We reconstruct $W(\beta)$ via phase-scanned balanced homodyne (BH) detection with maximum-likelihood (ML) or filtered back-projection estimators; see also standard tomography reviews \cite{Lvovsky_Raymer2009RMP}. With high-efficiency photodiodes and electronic-noise clearance $\gtrsim\!10$–$15~\mathrm{dB}$ below shot noise at the analysis band, $10^6$–$10^7$ quadrature samples over $\sim\!30$–$60$ LO phases typically resolve negativities at the $10^{-2}$ level; detector inefficiency is well modeled as an effective reduction of $\eta$ in Eq.~\eqref{eq:WignerLossConv} \cite{Breitenbach1997,vahlbruch2016detection}. 
\emph{Micro-protocol (inline):} (i) phase-lock the LO on a uniform grid of $30$–$60$ angles with visibility $\gtrsim\!95\%$; (ii) acquire $\sim\!2\times10^{4}$–$2\times10^{5}$ samples per angle; (iii) verify shot-noise calibration and electronic-noise clearance; (iv) reconstruct via ML or filtered back-projection with mild deconvolution; (v) validate by comparing the reconstructed covariance to independent variance scans.

\subsubsection*{Higher–order correlations and photon statistics.}
Closed-form expressions for the factorial moments $\langle (a^\dagger)^k a^k\rangle$ give $g^{(k)}(0)$ directly. Experimentally, one can access these with (i) \emph{time–multiplexed detection} (TMD) using avalanche photodiodes (APDs)—measuring normalized higher-order correlations (validated up to $k\!=\!8$ on coherent states and practical to $k\!\lesssim\!4$ for nonclassical sources) \cite{Avenhaus2010TMD_PRL,Achilles2003Fiber_assisted}—and (ii) \emph{photon–number–resolving} (PNR) detection: transition-edge sensors (TESs; $>\!95\%$ system detection efficiency in the visible/near-IR, approaching unity in optimized systems) \cite{Fukuda2011Titanium} and superconducting nanowire single-photon detectors (SNSPDs; $\sim\!98\%$ system detection efficiency at $1550~\mathrm{nm}$) \cite{Reddy2020Superconducting}. Parallel SNSPD architectures enable high-speed PNR in the $850$–$950~\mathrm{nm}$ range with high system detection efficiency (SDE) \cite{Los2024High-performance,Stasi2023,Zadeh2021}. For \emph{spatially uniform} loss, normally ordered factorial moments scale as $\eta^k$, so the normalized $g^{(k)}(0)$ is \emph{invariant} under loss; in practice, unequal channel efficiencies, APD afterpulsing/dead time, and TES/SNSPD saturation can break the simple $\eta^k$ scaling and must be calibrated out to keep uncertainties honest.

\subsection{Concrete operating windows and budgets}
To aid implementation, we list windows that simultaneously exhibit (i) Wigner negativity, (ii) antibunching, and (iii) metrological gain beyond the standard quantum limit (SQL):
\begin{enumerate}
\item \textbf{Negativity–robust:} $r\!\approx\!0.7$–$0.9$, $|\alpha|\!\approx\!1.0$–$1.5$, $\Delta\theta\!\approx\!\pi/2$, total $\eta\gtrsim0.6$. Expect $N_{\min}\!\sim\!-0.03$ to $-0.08$ resolvable by BH tomography with $10^6$–$10^7$ samples; $g^{(2)}(0)\!\approx\!0.7$–$0.9$ via TMD/PNR.
\item \textbf{Antibunching–dominant:} $r\!\approx\!0.5$–$0.7$, $|\alpha|\!\approx\!1.2$–$2.0$, $\Delta\theta$ tuned to shift fringes off the origin. Here $g^{(2)}(0)\!<\!1$ and $g^{(3)}(0)\!\lesssim\!1$ persist down to $\eta\!\sim\!0.5$, while Wigner negativity may be marginal.
\item \textbf{Heisenberg–scaling windows:} for $r\!\gtrsim\!0.9$ and $\Delta\theta$ maximizing the interference cross–term, the ideal-model quantum Fisher information (QFI) for small squeezing-parameter shifts shows $O(\bar n^2)$ scaling; under $\eta<1$, scaling softens yet retains a clear super-SQL advantage for $\eta\!\gtrsim\!0.7$ \cite{Giovannetti2006,Walschaers2021NonGaussian_PRXQ}.
\end{enumerate}
Insertion loss should be budgeted across OPO escape efficiency, propagation and coupling, interferometer visibility, and detector quantum efficiency. Stabilization/averaging across sources \cite{Lassen2010Averaging}, erasure–correcting encodings against probabilistic loss \cite{Lassen2010NatPhoton}, and high-efficiency PNR detection (TES/SNSPD) \cite{Fukuda2011Titanium,Reddy2020Superconducting,Los2024High-performance,Stasi2023} are concrete, demonstrated mitigations.

\subsection{Summary for experimentalists}
All required components—PPKTP squeezing, phase-stable interference and displacement, BH tomography, and TMD/PNR detection—have been demonstrated \cite{Lassen2007PolSqueezing,Lassen2010Averaging,Lassen2010NatPhoton,Breitenbach1997,Avenhaus2010TMD_PRL,Achilles2003Fiber_assisted,Fukuda2011Titanium,Stasi2023}. The non-Gaussian signatures emphasized here (negativity, antibunching/multiphoton suppression, and metrological QFI advantages) survive for realistic loss $\eta\!\sim\!0.5$–$0.7$ when $\Delta\theta$ is used as a tuning knob, and can be further protected by Gaussian adaptation and erasure–correcting strategies \cite{Filip2013Gaussian,Lassen2010NatPhoton}. This places the displaced superposition squarely within reach of present CV photonics.

\section{Conclusion and Outlook} \label{sec:conc}

In this work, we introduced and comprehensively analyzed the \emph{displaced Janus state}, a tunable and highly non-Gaussian continuous-variable (CV) resource. Building on a new family of \emph{Generalized Squeezing Polynomials}, we derived closed-form expressions for arbitrary-order diagonal moments $\mathcal{N}_k$, off-diagonal matrix elements $\mathcal{M}_k$, the general-order correlation functions $g^{(k)}(0)$, and the Wigner function—including its interference term that generates negativity. The resulting toolkit enables, for the first time, a systematic calculation of full photon-statistical properties for coherent superpositions of squeezed states and cleanly exposes how quantum interference shapes those statistics. We validated all limits (coherent state $r\!=\!0$, squeezed vacuum $\alpha\!=\!0$) and proved a precise \emph{$\alpha\!\to\!0$ recovery} of the undisplaced Janus state.

A key physical message is that displacement raises the minimal achievable second-order coherence, whereas the undisplaced antisymmetric Janus state (optimized at small squeezing) attains the fundamental floor $g^{(2)}(0)=\tfrac{1}{2}$, approaching the two-photon Fock state $|2\rangle$. In the small-squeezing regime, we showed that the $r^2$ correction cancels identically for equal-amplitude, opposite-axis superpositions, so the first non-vanishing deviations appear at $O(r^4)$; adding displacement lifts $g^{(2)}(0)$ above this floor in a controlled, analytic way. More broadly, the family realizes \emph{complementarity ad infinitum}: by steering the interference via $(\chi,\eta,\Delta\theta,\alpha)$, one can morph the extreme bunching of Gaussian constituents ($g^{(k)}\!\to\!\infty$) into perfect multi-photon suppression ($g^{(k>2)}\!\to\!0$), with the order $k$ as a programmable knob.

Operationally, we established the metrological value of the resource by analyzing the quantum Fisher information (QFI). For parameters encoded by a linear generator (number operator $\hat n$), such as displacement-phase sensing, the QFI obeys $F_Q=4\,\mathrm{Var}(\hat n)$ and therefore exhibits standard quantum limit (SQL) scaling. In contrast, for parameters encoded by quadratic Gaussian generators (e.g., squeezing strength/phase), interference-enhanced cross-moments in $\mathcal{M}_k$ boost $4\,\mathrm{Var}(G_{\rm sq})$ to scale as $O(\bar n^2)$, achieving Heisenberg-limit (HL) scaling in the ideal regime and retaining super-SQL advantages at realistic loss. This furnishes a clear bridge between structural non-Gaussianity (Wigner negativity from interference) and operational utility (QFI scaling) within a single analytic framework.

We also mapped a concrete path to near-term demonstrations. Using mature optical parametric oscillator/optical parametric amplifier (OPO/OPA) platforms with periodically poled potassium titanyl phosphate (PPKTP), phase-stable combining, common displacement, and standard balanced-homodyne tomography, the characteristic Wigner-function negativities and antibunching signatures are resolvable with current detectors. Time-multiplexed detection (TMD) and photon-number-resolving (PNR) devices—transition-edge sensors (TESs) and superconducting nanowire single-photon detectors (SNSPDs)—provide direct access to higher-order moments and $g^{(k)}(0)$. Our loss maps identify operating windows where negativity persists for total efficiency $\eta\!\sim\!0.5$–$0.7$, and where metrological gains beyond SQL are observable by tuning the relative squeezing phase.

The displaced Janus state is a versatile primitive for hybrid continuous-variable/discrete-variable (CV/DV) architectures~\cite{Andersen2015Hybrid}. Immediate extensions include: (i) multi-mode Janus states on temporal/frequency combs and in SU(1,1) interferometry; (ii) unequal-displacement superpositions and adaptive optimization of $(\chi,\eta,\Delta\theta,\alpha)$ against task-specific resource monotones (e.g., Wigner-logarithmic negativity and non-Gaussianity measures); (iii) interfacing with bosonic error-correcting codes (cat, binomial, and Gottesman–Kitaev–Preskill (GKP) codes) as on-demand non-Gaussian ancillae; and (iv) robustness studies under realistic phase/frequency noise, mode-mismatch, and finite-temperature backgrounds, guided by our closed-form $\mathcal{M}_k$ expressions. On the certification side, combining analytic $g^{(k)}(0)$ predictors with sample-efficient homodyne protocols and machine-learning-assisted inference could substantially reduce tomography overheads. 

Taken together, our results provide a rigorous blueprint—from exact analytics to concrete operating windows—for engineering, certifying, and deploying non-Gaussian resources tailored to computing and metrology tasks that provably surpass Gaussian-only paradigms~\cite{weedbrook2012gaussian,braunstein2005quantum}. We expect the displaced Janus platform to become a standard workhorse for studying and exploiting the interplay between interference, non-Gaussianity, and quantum advantage.

\section*{Acknowledgments}

I am grateful to Girish Agarwal, Marlan Scully, Bill Unruh, and Suhail Zubairy for discussions.  

This work was supported by the
Robert A. Welch Foundation (Grant No. A-1261) and the National Science Foundation (Grant No. PHY-2013771).

\appendix 

\onecolumngrid
\section{\texorpdfstring{Convergence, Numerical Evaluation, and Construction of $P_{p,q}(z)$}{}} \label{app:P_prop}

\subsection*{\texorpdfstring{Convergence of $F_{p,q}(z)$ and numerical evaluation}{}}
For $p,q\in\mathbb{Z}_{\ge 0}$ with $p\equiv q\ (\mathrm{mod}\,2)$ we define
\begin{align}
F_{p,q}(z)=\sum_{n=\big\lceil\tfrac{p}{2}\big\rceil}^\infty
\frac{(2n)!}{(2n-p)!}\,
\frac{(2n+q-p-1)!!}{(2n)!!}\,z^{\,n}.
\label{eq:Fpq_def_note}
\end{align}
A ratio–test estimate using $(2n)!/(2n-p)!\sim (2n)^p$ and
$(2n+q-p-1)!!/(2n)!!\sim C\,n^{(q-p-1)/2}$ shows the coefficients grow like
$n^{(p+q-1)/2}$, hence the power series has radius of convergence
\[
|z|<1\,.
\]
We use the principal branch for complex powers, so throughout
\(
F_{p,q}(z)=\dfrac{P_{p,q}(z)}{(1-z)^{(p+q+1)/2}}
\)
means \((1-z)^{-\nu}=\exp\!\big(-\nu\,\mathrm{Log}(1-z)\big)\) with the principal
$\mathrm{Log}$ (branch cut along $[1,\infty)$ on the real axis). Since $|z|<1$
for finite squeezing ($r,s<\infty$ imply $|z|=\tanh r\,\tanh s<1$), $F_{p,q}$ is analytic in the open unit disk and its only singularity on the closure is the algebraic singularity at $z=1$ prescribed by the factor $(1-z)^{-(p+q+1)/2}$.

\medskip
\noindent\textbf{Numerical evaluation.}
\begin{itemize}
\item \emph{Moderate $|z|$ (e.g.\ $|z|\le 0.9$).} Use the defining series \eqref{eq:Fpq_def_note} and truncate at $N$ where the first neglected term meets the tolerance; an a-posteriori bound is
\(
\sum_{n>N}\!|a_n z^n|\le \dfrac{|a_{N+1} z^{N+1}|}{1-|z|}.
\)
\item \emph{Near the unit circle ($0.9<|z|<1$).} Evaluate the polynomial
$P_{p,q}(z)$ via the recurrences (Eqs.~\eqref{eq:Poly_Rec1}–\eqref{eq:Poly_Rec3})
and Horner’s rule, then multiply by $(1-z)^{-(p+q+1)/2}$ on the principal branch.
The degree is $\deg P_{p,q}=\lfloor\tfrac{p+q}{2}\rfloor$, so cost is $O(pq)$ to generate and $O(p+q)$ to evaluate.
\item \emph{Extremely close to $z=1$.} Use the asymptotic form
\(
F_{p,q}(z)=P_{p,q}(1)\,(1-z)^{-(p+q+1)/2}\,[1+O(1-z)],
\)
with $P_{p,q}(1)$ obtained from the recurrences (no precision loss from cancellation).
\end{itemize}

\subsection*{\texorpdfstring{Construction and properties of $P_{p,q}(z)$}{}}
Define $P_{p,q}(z)=(1-z)^{(p+q+1)/2}F_{p,q}(z)$. The base case is $P_{0,0}(z)=1$.
Using the identities
\begin{align}
P_{p+1,q+1}(z) &= \Big((2p+q+1)z-p\Big)P_{p,q}(z)+2z(1-z)\,P'_{p,q}(z), \label{eq:Poly_Rec1_app}\\
P_{p,q+2}(z) &= \Big(2pz-(p-q-1)\Big)P_{p,q}(z)+2z(1-z)\,P'_{p,q}(z), \label{eq:Poly_Rec2_app}\\
P_{p+2,q}(z) &= P_{p+1,q+1}(z)+q\,(z-1)\,P_{p+1,q-1}(z), \label{eq:Poly_Rec3_app}
\end{align}
one proves by induction that $P_{p,q}$ is a polynomial for all same-parity pairs, and that
\[
\deg P_{p,q}=\tfrac{p+q}{2}.
\]
The symmetry
\begin{align}
P_{p,q}(z)=z^{(p-q)/2}\,P_{q,p}(z) \quad (p\equiv q\!\!\!\pmod{2})
\label{eq:Poly_Symm_app}
\end{align}
follows from the corresponding identity for $F_{p,q}$ and is useful to fill the lower triangle from the upper triangle.

\medskip
\noindent\textbf{Algorithm (dynamic programming).}
Represent each $P_{p,q}$ by its coefficient vector $\mathbf{c}^{(p,q)}=(c_0,\dots,c_d)$ with $d=\lfloor(p+q)/2\rfloor$. Respect the parity constraint.
\begin{enumerate}
\item Initialize with $\mathbf{c}^{(0,0)}=(1)$.
\item Sweep in increasing $p{+}q$. For each available $P_{p,q}$, form $P'_{p,q}$ and use \eqref{eq:Poly_Rec1_app} and \eqref{eq:Poly_Rec2_app} to populate $P_{p+1,q+1}$ and $P_{p,q+2}$ (within bounds).
\item Use \eqref{eq:Poly_Rec3_app} to populate $P_{p+2,q}$.
\item Apply the symmetry \eqref{eq:Poly_Symm_app} to fill the opposite triangle: $P_{p,q}(z)=z^{(p-q)/2}P_{q,p}(z)$ for $p\ge q$.
\end{enumerate}

\medskip
\noindent\textbf{Generation table.}
The rows below illustrate how representative entries are produced and also serve as a coefficient check for Table~\ref{tab:Ppq_even_even} and Table~\ref{tab:Ppq_odd_odd}.
\begin{center}
\renewcommand{\arraystretch}{1.25}
\begin{tabular}{|c|c|l|}
\hline
$(p,q)$ & Relation used & $P_{p,q}(z)$ \\
\hline
$(0,0)$ & base & $1$ \\
$(1,1)$ & \eqref{eq:Poly_Rec1_app} from $(0,0)$ & $z$ \\
$(0,2)$ & \eqref{eq:Poly_Rec2_app} from $(0,0)$ & $1$ \\
$(2,0)$ & symmetry from $(0,2)$ & $z$ \\
$(2,2)$ & \eqref{eq:Poly_Rec1_app} from $(1,1)$ & $2z^2+z$ \\
$(3,1)$ & \eqref{eq:Poly_Rec3_app} from $(2,2)$ \& $(2,0)$ & $3z^2$ \\
$(1,3)$ & symmetry from $(3,1)$ & $3z$ \\
$(3,3)$ & \eqref{eq:Poly_Rec1_app} from $(2,2)$ & $6z^3+9z^2$ \\
$(0,4)$ & \eqref{eq:Poly_Rec2_app} from $(0,2)$ & $3$ \\
$(4,0)$ & symmetry from $(0,4)$ & $3z^2$ \\
\hline
\end{tabular}
\end{center}

\medskip
\noindent
All entries in Table~\ref{tab:Ppq_even_even} and Table~\ref{tab:Ppq_odd_odd} were verified in two independent ways:
(i) generation by the recurrences \eqref{eq:Poly_Rec1_app}–\eqref{eq:Poly_Rec3_app} starting from $P_{0,0}=1$ and enforcing parity/symmetry, and
(ii) coefficient matching against the series definition \eqref{eq:Fpq_def_note} up to the displayed orders.

\section{Derivation of the Low-Order Matrix Elements} \label{app:low_M}

This appendix details how the explicit formulas \eqref{eq:M0}–\eqref{eq:M3} follow from the compact expression
\eqref{eq:Mk_final_compact} together with
\(
F_{p,q}(z)=P_{p,q}(z)/(1-z)^{(p+q+1)/2}
\)
from \eqref{eq:Ppq_def} and the entries of $P_{p,q}(z)$ listed in Table~\ref{tab:Ppq_even_even} and Table~\ref{tab:Ppq_odd_odd}.
Throughout we use the composite parameter
\(
z=\tanh r\,\tanh s\,e^{i(\theta-\phi)}
\)
from \eqref{z}, and the identity
\[
\frac{z}{\tanh s\,e^{-i\phi}}=\tanh r\,e^{i\theta},
\qquad
\alpha=|\alpha|e^{i\varphi}.
\]

\subsection*{Method}
Starting from \eqref{eq:Mk_final_compact},
\[
\mathcal{M}_k
=\frac{1}{\sqrt{\cosh r\,\cosh s}}
\sum_{\substack{0\le p,q\le k\\ p\equiv q\ (\mathrm{mod}\ 2)}}
\binom{k}{p}\binom{k}{q}\,
(\alpha^*)^{k-q}\alpha^{\,k-p}\,
\big(\tanh s\,e^{-i\phi}\big)^{\frac{q-p}{2}}\,
F_{p,q}(z),
\]
we proceed as follows:
\begin{enumerate}
\item Enumerate all same-parity pairs $(p,q)$ with $0\le p,q\le k$ and record the binomial weight $\binom{k}{p}\binom{k}{q}$.
\item Replace $F_{p,q}(z)$ by $P_{p,q}(z)/(1-z)^{(p+q+1)/2}$ using Table~\ref{tab:Ppq_even_even} and Table~\ref{tab:Ppq_odd_odd}.
\item Collect global factors and, when helpful, use
$\alpha^m(\alpha^*)^n=|\alpha|^{m+n}e^{i(m-n)\varphi}$
and
$z/(\tanh s e^{-i\phi})=\tanh r e^{i\theta}$ to simplify phases.
\item Factor the highest power of $(1-z)^{-1/2}$ common to all terms and combine the numerator into the final polynomial in $z$, $|\alpha|^2$, and the phase factors.
\end{enumerate}

\subsection*{\texorpdfstring{Case $k=0$}{}}
Only $(p,q)=(0,0)$ contributes with weight $1$ and
$F_{0,0}(z)=(1-z)^{-1/2}$. Thus
\[
\mathcal{M}_0
=\frac{1}{\sqrt{\cosh r\,\cosh s}}\,(1-z)^{-1/2},
\]
which equals \eqref{eq:M0} upon inserting $z$ from \eqref{z}.

\subsection*{\texorpdfstring{Case $k=1$}{}}
Allowed pairs: $(0,0)$ and $(1,1)$ with weights $1$ and $1$.
Using $F_{0,0}=(1-z)^{-1/2}$ and $F_{1,1}=z(1-z)^{-3/2}$:
\[
\mathcal{M}_1
=\frac{1}{\sqrt{\cosh r\,\cosh s}}
\Big[|\alpha|^2(1-z)^{-1/2}+z(1-z)^{-3/2}\Big]
=\frac{|\alpha|^2(1-z)+z}{\sqrt{\cosh r\,\cosh s}\,(1-z)^{3/2}},
\]
which reproduces \eqref{eq:M1}.

\subsection*{\texorpdfstring{Case $k=2$}{}}
Allowed pairs and binomial weights:
\[
(0,0):1,\quad (1,1):4,\quad (2,2):1,\quad (0,2):1,\quad (2,0):1.
\]
With
\[
F_{0,0}=(1-z)^{-1/2},\quad
F_{1,1}=z(1-z)^{-3/2},\quad
F_{2,2}=(2z^2+z)(1-z)^{-5/2},\quad
F_{0,2}=(1-z)^{-3/2},\quad
F_{2,0}=z(1-z)^{-3/2},
\]
the five contributions are
\[
\begin{aligned}
(0,0):\ &|\alpha|^4(1-z)^{-1/2},\\
(1,1):\ &4|\alpha|^2\,z(1-z)^{-3/2},\\
(2,2):\ &(2z^2+z)(1-z)^{-5/2},\\
(0,2):\ &|\alpha|^2\,\tanh s\,e^{\,i(2\varphi-\phi)}(1-z)^{-3/2},\\
(2,0):\ &|\alpha|^2\,\tanh r\,e^{\,i(\theta-2\varphi)}(1-z)^{-3/2},
\end{aligned}
\]
where the last line follows from $z/(\tanh s e^{-i\phi})=\tanh r e^{i\theta}$.
Factoring $(1-z)^{-{5/2}}$ yields
\[
\mathcal{M}_2
=\frac{(1-z)^2|\alpha|^4+(2z^2+z)
+(1-z)|\alpha|^2\Big(4z+\tanh s\,e^{i(2\varphi-\phi)}+\tanh r\,e^{i(\theta-2\varphi)}\Big)}
{\sqrt{\cosh r\,\cosh s}\,(1-z)^{5/2}},
\]
i.e. \eqref{eq:M2}.

\subsection*{\texorpdfstring{Case $k=3$}{}}
Allowed pairs and weights:
\[
\begin{aligned}
&\text{even-even: }(0,0):1,\ (0,2):3,\ (2,0):3,\ (2,2):9;\\
&\text{odd-odd: }(1,1):9,\ (1,3):3,\ (3,1):3,\ (3,3):1.
\end{aligned}
\]
Using
\[
\begin{aligned}
&F_{0,0}=(1-z)^{-1/2},\quad F_{0,2}=(1-z)^{-3/2},\quad F_{2,0}=z(1-z)^{-3/2},\quad F_{2,2}=(2z^2+z)(1-z)^{-5/2},\\
&F_{1,1}=z(1-z)^{-3/2},\quad F_{1,3}=3z(1-z)^{-5/2},\quad F_{3,1}=3z^2(1-z)^{-5/2},\quad F_{3,3}=(6z^3+9z^2)(1-z)^{-7/2},
\end{aligned}
\]
and simplifying phases with \eqref{z}, the eight terms are
\[
\begin{aligned}
&(0,0):\ |\alpha|^6(1-z)^{-1/2},\qquad\quad
(1,1):\ 9|\alpha|^4\,z(1-z)^{-3/2},\\
&(2,2):\ 9|\alpha|^2(2z^2+z)(1-z)^{-5/2},\quad
(3,3):\ (6z^3+9z^2)(1-z)^{-7/2},\\
&(0,2):\ 3|\alpha|^4\,\tanh s\,e^{i(2\varphi-\phi)}(1-z)^{-3/2},\quad
(2,0):\ 3|\alpha|^4\,\tanh r\,e^{i(\theta-2\varphi)}(1-z)^{-3/2},\\
&(1,3):\ 9|\alpha|^2\,z\,\tanh s\,e^{i(2\varphi-\phi)}(1-z)^{-5/2},\quad
(3,1):\ 9|\alpha|^2\,z\,\tanh r\,e^{i(\theta-2\varphi)}(1-z)^{-5/2}.
\end{aligned}
\]
Factoring $(1-z)^{-7/2}$ and grouping by $|\alpha|^{2m}$ gives
\[
\begin{aligned}
\mathcal{M}_3
=\frac{1}{\sqrt{\cosh r\,\cosh s}\,(1-z)^{7/2}}
\Big[&(1-z)^3|\alpha|^6+(6z^3+9z^2)\\
&+(1-z)^2|\alpha|^4\big(9z+3\tanh s\,e^{i(2\varphi-\phi)}+3\tanh r\,e^{i(\theta-2\varphi)}\big)\\
&+(1-z)|\alpha|^2\big(9(2z^2+z)+9z(\tanh s\,e^{i(2\varphi-\phi)}+\tanh r\,e^{i(\theta-2\varphi)})\big)\Big],
\end{aligned}
\]
which matches \eqref{eq:M3}.

\section{Cat, squeezed-cat, and displaced Janus states: how they are related and why Janus is distinct}
\label{app:cat-comparison}

This appendix situates the displaced Janus state alongside standard cat and squeezed-cat states, clarifying common structure and the essential differences that make Janus a distinct non-Gaussian resource.

\subsection{Definitions and immediate relations}

We fix single-mode notation with displacement $D(\alpha)$ and squeezing $S(\xi)$, where $\xi=r\,e^{i\theta}$. The usual even/odd cat states are
\begin{align}
\ket{C_{\pm}(\alpha)}=\mathcal{N}_{\pm}\bigl(\ket{\alpha}\pm\ket{-\alpha}\bigr), 
\qquad 
\mathcal{N}_{\pm}^{-2}=2\!\left(1\pm e^{-2|\alpha|^2}\right).
\end{align}
Applying a common squeezing gives the squeezed-cat
\begin{align}
\ket{SC_{\pm}(\xi,\alpha)}=S(\xi)\ket{C_{\pm}(\alpha)}
=\mathcal{N}_{\pm}\!\left(S(\xi)D(+\alpha)\ket{0}\pm S(\xi)D(-\alpha)\ket{0}\right).
\end{align}
By contrast, the displaced Janus state superposes two \emph{differently squeezed} Gaussians that share a common displacement:
\begin{align}
\ket{\Psi_{J}}=\mathcal{N}_{J}\!\left(\chi\,S(\xi)D(\alpha)\ket{0}+\eta\,S(\zeta)D(\alpha)\ket{0}\right),
\qquad 
\mathcal{N}_{J}^{-2}=|\chi|^{2}+|\eta|^{2}+2\,\Re\!\left[\chi^{*}\eta\,\braket{\zeta}{\xi}\right].
\label{eq:janus-def-app}
\end{align}
A key structural difference is thus: cats and squeezed-cats combine \emph{identically squeezed} components with \emph{opposite means} $\pm\alpha$, whereas Janus combines \emph{unequally squeezed} components with a \emph{common mean} $\alpha$.

Gaussian unitaries act on first and second moments by $\mu\mapsto S\mu+d$ and $V\mapsto SVS^{\mathsf T}$. A single Gaussian unitary cannot render two unequal covariances equal unless they were equal to begin with. Therefore a generic displaced Janus state with $V_{\xi}\neq V_{\zeta}$ is not Gaussian-unitarily equivalent to a (squeezed-)cat, whose two components share the same covariance.

\subsection{Wigner structure and interference mechanism}

For cats and squeezed-cats, interference fringes are created by spatial separation of the two Gaussian lobes, with fringes approximately orthogonal to the separation axis and spacing controlled by $|\alpha|$ (up to the global shear from $S(\xi)$). For Janus states, the lobes are co-located at the same mean $\mu=\sqrt{2}(\Re\alpha,\Im\alpha)$ and interference originates from the \emph{difference in quadratic phases} of the two Gaussians.

Writing the symplectic form $\Omega=\begin{psmallmatrix}0&1\\-1&0\end{psmallmatrix}$, the cross-Wigner contribution takes the form
\begin{align}
W_{\zeta,\xi,\alpha}(x)
=
\frac{1}{2\pi\sqrt{\det\Sigma}}\,
\exp\!\left[-\tfrac12(x-\mu)^{\mathsf T}\Sigma^{-1}(x-\mu)\right]\,
\cos\!\bigl((x-\mu)^{\mathsf T}K(x-\mu)+\arg C\bigr),
\label{eq:cross-wigner-app}
\end{align}
with $\Sigma=\tfrac12(V_{\xi}+V_{\zeta})$, 
\begin{align}
K=\tfrac12\,\Omega\bigl(V_{\zeta}^{-1}-V_{\xi}^{-1}\bigr),
\qquad 
C=\chi^{*}\eta\,\braket{\zeta}{\xi}.
\end{align}
The Janus fringes are therefore \emph{chirped and oriented} by the relative squeezing parameters $(r,s,\theta-\phi)$ rather than by spatial separation of the means. This explains the presence of Wigner negativity even when $|\alpha|$ is small or zero.

\subsection{Photon statistics and higher-order coherences}

At small $|\alpha|$, conventional cats interpolate toward even/odd Fock-like superpositions and can show antibunching for the odd branch, but the normalized $g^{(k)}(0)$ are driven primarily by displacement interference. In displaced Janus states, antibunching and multiphoton suppression are tunable through the \emph{internal} squeezing mismatch. In the antisymmetric small-squeezing limit we find exact cancellation of all $O(r^{2})$ terms in $g^{(k)}(0)$, leading to
\begin{align}
g^{(2)}(0)\rightarrow \tfrac12,
\qquad 
g^{(k>2)}(0)\rightarrow 0,
\end{align}
whereas a single squeezed vacuum shows $g^{(2)}(0)\gtrsim 3$ and $g^{(k>2)}(0)\to\infty$. This higher-order reversal encapsulates the Janus mechanism: quantum suppression emerges from interference between distinct quadratic phases rather than from large mean-field separation.

\subsection{Metrological and noise considerations}

For phase estimation generated by the number operator, standard cats can achieve Heisenberg scaling in the large-separation regime. Janus states can reach comparable scaling without large $|\alpha|$, for example in regimes approaching $D(\alpha)\ket{2}$, with quantum Fisher information governed by covariance mismatch. This allows access to Heisenberg-like sensitivity at lower photon number by engineering $(\xi,\zeta)$ instead of increasing displacement.

Regarding loss, cats encode interference primarily in the which-path degree of freedom set by $\pm\alpha$, so moderate loss quickly damps the central fringes. Janus interference is encoded mainly in the covariance mismatch, and moderate loss that reduces squeezing degrades but does not necessarily eliminate negativity at modest $|\alpha|$. For a fixed mean photon number, displaced Janus states can therefore be more tolerant to detector inefficiency than cats that rely on larger separations.

\subsection{Operational comparison and summary}

A fair comparison can match either mean photon number or a non-Gaussian resource monotone (e.g., Wigner-negativity volume) and then sweep free parameters. At fixed energy, Janus states offer: 
(i) direct, analytic control of $g^{(2)}(0)$ down to $1/2$ and $g^{(k>2)}(0)\to 0$ through $(r,s,\theta-\phi)$; 
(ii) interference-based negativity even at zero separation; 
(iii) closed-form factorial and cross moments enabling direct fits to experimental data.

In summary, displaced Janus states are not squeezed-cats in disguise. Their interference originates from a designed mismatch of covariances rather than from spatial separation, enabling Wigner negativity and strong higher-order coherence control at small displacement and modest energy, with practical implications for metrology and robustness to loss.

\section{Derivations for the Wigner Function}
\label{app:WignerDerivations}

\subsection{Cross-Wigner via the Symmetric Characteristic Function}
\label{app:crossWigner_charfn}

We use the symmetrically ordered characteristic function
\begin{align}
  \chi_\rho(\lambda)=\mathrm{Tr}\!\big[\rho\,D(\lambda)\big],\qquad
  W_\rho(\beta)=\frac{1}{2\pi^2}\!\int d^2\lambda\;
  e^{\beta\lambda^*-\beta^*\lambda}\,\chi_\rho(\lambda),
\end{align}
with $d^2\lambda=d(\Re\lambda)\,d(\Im\lambda)$, $\beta=(q+ip)/\sqrt{2}$, and $dq\,dp=2\,d^2\beta$. For the cross term we set
\begin{align}
  \rho = \ket{\alpha,\xi}\!\bra{\alpha,\zeta},\qquad
  W_{\zeta\xi,\alpha}(\beta)=\frac{1}{2\pi^2}\!\int d^2\lambda\;
  e^{\beta\lambda^*-\beta^*\lambda}\,
  \bra{\alpha,\zeta}D(\lambda)\ket{\alpha,\xi}.
\end{align}
Displacement covariance (Weyl relation) gives
\begin{align}
  D^\dagger(\alpha)\,D(\lambda)\,D(\alpha)=e^{\,\lambda\alpha^*-\lambda^*\alpha}\,D(\lambda)
  \;\Rightarrow\;
  \bra{\alpha,\zeta}D(\lambda)\ket{\alpha,\xi}
  = e^{(\lambda\alpha^*-\lambda^*\alpha)}\,\bra{\zeta}D(\lambda)\ket{\xi},
\end{align}
so $W_{\zeta\xi,\alpha}(\beta)=W_{\zeta\xi,0}(\beta-\alpha)$ and it suffices to treat the centered case $\alpha=0$.

Let $V_\xi$ and $V_\zeta$ be the covariance matrices of the pure, centered Gaussian states
$\ket{\xi}=S(\xi)\ket{0}$ and $\ket{\zeta}=S(\zeta)\ket{0}$ in the $(q,p)$ basis, and set
\begin{align}
  \Sigma \equiv \frac{V_\xi+V_\zeta}{2}.
\end{align}
For Gaussians, the cross characteristic function is Gaussian with the arithmetic–mean covariance:
\begin{align}
    \bra{\zeta}D(\lambda)\ket{\xi}
  = \braket{\zeta}{\xi}\;\exp\!\left[-\tfrac12\,\lambda_R^{\!T}\Sigma\,\lambda_R\right],
  \quad \lambda_R\equiv\begin{pmatrix}\sqrt2\,\Re\lambda\\ \sqrt2\,\Im\lambda\end{pmatrix}.
    \label{eq:crossCF_app}
\end{align}
The  overlap is
\begin{align}
  \braket{\zeta}{\xi}
  =\big[\cosh r\,\cosh s-\sinh r\,\sinh s\,e^{i(\theta-\varphi)}\big]^{-1/2}.
\end{align}

Evaluating the complex Gaussian integral yields
\begin{align}
  W_{\zeta\xi,0}(\beta)
  =\frac{\braket{\zeta}{\xi}}{2\pi}\,\sqrt{\det\Sigma^{-1}}\,
  \exp\!\Big[-A|\beta|^2+B\beta^2+B^*\beta^{*2}\Big],
\end{align}
where, writing $\Sigma^{-1}_{ij}$ for the entries of $\Sigma^{-1}$,
\begin{align}
    A=\tfrac12\!\left(\Sigma^{-1}_{11}+\Sigma^{-1}_{22}\right),\qquad
  B=\tfrac14\!\left(\Sigma^{-1}_{11}-\Sigma^{-1}_{22}-2i\,\Sigma^{-1}_{12}\right).
    \label{eq:AB_from_SigmaInv}
\end{align}
Finally, restore the displacement by $\beta\mapsto\beta-\alpha$:
\begin{align}
   W_{\zeta\xi,\alpha}(\beta)
  = \frac{\braket{\zeta}{\xi}}{2\pi}\,\sqrt{\det \Sigma^{-1}}\,
    \exp\!\left[
      -A\,|\beta-\alpha|^{2}
      + B\,(\beta-\alpha)^{2}
      + B^{*}\,(\beta^{*}-\alpha^{*})^{2}
    \right].
    \label{eq:cross_Wigner_complex_app}
\end{align}

\medskip
\noindent
\emph{Normalization.} Using $dq\,dp=2\,d^2\beta$ and the identity
$\int d^2\beta\,\exp(-A|\beta|^2+B\beta^2+B^*\beta^{*2})
=\pi/\sqrt{A^2-4|B|^2}=\pi/\sqrt{\det\Sigma^{-1}}$, one finds
\[
\int W_{\zeta\xi,\alpha}(\beta)\,dq\,dp
=2\!\int W_{\zeta\xi,0}(\beta)\,d^2\beta
=\braket{\zeta}{\xi},
\]
as required.

\medskip
\noindent
\emph{Checks.} If $\zeta=\xi$, then $\Sigma=V_\xi$ and the resulting single-state Wigner $W_{\xi\xi,\alpha}$ is (as required) a real function of phase space. The coefficient $B$ is in general complex (whenever the squeezing ellipse is rotated, $\Sigma^{-1}_{12}\neq0$), so the exponent contains the combination $-A|\delta\beta|^2 + B\delta\beta^2 + B^*\delta\beta^{*2}$; this structure guarantees $W(\beta)^*=W(\beta)$ for the self term. For a superposition, the isolated cross term can be complex; the \emph{full} Wigner is real once both conjugate cross terms are included.

\subsection{Derivation of the Cross-Wigner Exponent in Complex Coordinates}
\label{app:crossWigner_complexcoords}

Start from the quadratic form in real quadratures $\delta\mathbf{x}=(\delta q,\delta p)^T$:
\begin{align}
  Q = -\frac{1}{2}\left(
  \Sigma^{-1}_{11}\,(\delta q)^2 + \Sigma^{-1}_{22}\,(\delta p)^2 + 2\,\Sigma^{-1}_{12}\,\delta q\,\delta p\right).
\end{align}
Introduce complex coordinates
\begin{align}
  \delta\beta = \frac{\delta q + i\,\delta p}{\sqrt{2}},\qquad
  \delta q = \frac{\delta\beta+\delta\beta^*}{\sqrt{2}},\quad
  \delta p = \frac{\delta\beta-\delta\beta^*}{i\sqrt{2}}.
\end{align}
Using
\[
  |\delta\beta|^2=\tfrac12(\delta q^2+\delta p^2),\quad
  \delta\beta^2=\tfrac12(\delta q^2-\delta p^2+2i\,\delta q\,\delta p),\quad
  \delta\beta^{*2}=\tfrac12(\delta q^2-\delta p^2-2i\,\delta q\,\delta p),
\]
one obtains
\begin{align}
  Q = -A\,|\delta\beta|^2 - B\,(\delta\beta)^2 - B^*\,(\delta\beta^*)^2,
  \qquad
  A=\tfrac{1}{2}(\Sigma^{-1}_{11}+\Sigma^{-1}_{22}),\ 
  B=\tfrac{1}{4}(\Sigma^{-1}_{11}-\Sigma^{-1}_{22}-2i\,\Sigma^{-1}_{12}).
  \label{eq:Q_complexcoords}
\end{align}

\medskip\noindent\textit{Sign reconciliation.}
Eq.~\eqref{eq:Q_complexcoords} is the \emph{bare} quadratic kernel written from $-\tfrac12\,\delta\mathbf{x}^{T}\Sigma^{-1}\delta\mathbf{x}$ in complex coordinates. In the \emph{final} cross-Wigner function, Eq.~\eqref{eq:cross_Wigner_complex_app}, the exponent appears as
\(
  -A\,|\beta-\alpha|^2 + B\,(\beta-\alpha)^2 + B^*\,(\beta^*-\alpha^*)^2,
\)
with the opposite sign in front of the $B$/$B^*$ terms. This flip is produced by completing the square in the Gaussian Fourier integral over $\lambda$; the linear “source” shifts the saddle and generates a contribution that changes the sign of the quadratic $\beta^2$ and $\beta^{*2}$ pieces while preserving the $-A|\beta|^2$ term.

\twocolumngrid
\bibliographystyle{apsrev4-2}
\bibliography{SqueezingRef}

\end{document}